\documentclass[a4paper,11pt]{article}
\usepackage{fullpage}

\usepackage[persection]{belov_paper}
\numberwithin{equation}{section}

\def\k#1{\ket|#1>}

\usepackage{tikz}
\usetikzlibrary{calc,graphs,positioning,decorations.pathmorphing,arrows.meta}
\input{xy}

\usepackage{mathabx}
\let\succeq\succcurlyeq
\let\preceq\preccurlyeq

\def\SkobInner{\LLL\vvvert##1\RRR\vvvert} \makeSkob{Dnorm}||

\usepackage[framemethod=TikZ]{mdframed}

\newmdenv[%
  roundcorner=5pt,
  linecolor=blue!15,
  linewidth=2pt,
  subtitlebackgroundcolor=blue!15,
  subtitleaboveskip=0pt,
  subtitlebelowskip=0pt,
  subtitleinneraboveskip=0pt,
  innerbottommargin=0pt,
  subtitlefont=\normalfont
]{mdfigure}

\def\myfigureInternal#1#2#3{
\refstepcounter{figure} #1
\begin{mdfigure}[
  frametitle={
    \tikz[baseline=(current bounding box.east),outer sep=0pt]
    \node[anchor=east,rectangle,fill=blue!15,rounded corners]
    {Figure~\thefigure};},
  frametitleaboveskip=-10pt
  innertopmargin=0pt,
  innerbottommargin=0pt,
  roundcorner=5pt,
  linecolor=blue!15,
  linewidth=2pt,
  subtitlebackgroundcolor=blue!15,
  subtitleaboveskip=0pt,
  subtitlebelowskip=0pt,
  subtitleinneraboveskip=0pt,
  subtitlefont=\normalfont
]
#3
\mdfsubtitle{\medskip #2}
\end{mdfigure}
}

\def\myfigure#1#2#3{
\begin{figure}[htb]
\myfigureInternal{#1}{#2}{#3}
\negbigskip
\end{figure}
}

\newcommand{\oxplus}{\oplus}
\newcommand{\bigoxplus}{\bigoplus}

\usepackage[utf8x]{inputenc}
\usepackage[greek.polutoniko, english]{babel}

\begin{document}

\title{One-Way Ticket to Las Vegas and the Quantum Adversary}
\author{Aleksandrs Belovs\thanks{Faculty of Computing, University of Latvia} 
\and 
Duyal Yolcu\thanks{\url{https://github.com/qudent}}}
\date{}
\maketitle

\mycutecommand\onegamma{{\mathop{\gamma_2}\limits^\leftarrow}}
\mycutecommand\twogamma{{\mathop{\gamma_2}\limits^\leftrightarrow}}
\mycutecommand\subgamma{{\mathop{\gamma_2}\limits^\curvearrowbotleft}}


\begin{abstract}
We propose a new definition of quantum Las Vegas query complexity. 
We show that it is exactly equal to the quantum adversary bound.
This is achieved by a new and very simple way of transforming a feasible solution to the adversary optimisation problem into a quantum query algorithm.
This allows us to generalise the bound to include unidirectional access, multiple input oracles, and input oracles that are not unitary.
As an application, we demonstrate a separation between unidirectional and bidirectional access to an input oracle for a rather natural unitary permutation inversion problem.
\end{abstract}

\section{Introduction}
This paper combines two topics: Las Vegas query complexity and the quantum adversary bound.

\subsection{Las Vegas Complexity}
There are two main types of \emph{randomised} query algorithms, with different complexity measures:
\begin{itemize}
\item A \emph{Monte Carlo} query algorithm, also known as bounded-error, is allowed to output an incorrect answer with some small probability $\eps$, usually $1/3$.
The algorithm is allowed to make a certain number of queries, which cannot be exceeded.
This number is the query complexity of the algorithm.
\item A \emph{Las Vegas} query algorithm, also known as zero-error, always has to give the correct answer.
On the other hand, it has no strict limit on the number of queries it can make.
Sometimes it can make few queries, sometimes a lot.
Its query complexity is defined as the \emph{expected} number of queries it makes on a certain input.
\end{itemize}

Las Vegas algorithms have a number of nice properties.
First, complexity is independent of the choice of the error parameter $\eps$.  
Hence, one can talk about \emph{the exact} value of complexity for a particular problem on a particular input, which can even not be an integer.
Second, Las Vegas algorithms can be nicely composed as there is no need for error reduction. 
For Monte Carlo algorithms, one usually gets extra logarithmic factors due to the necessity to reduce the error of the inner subroutines.

One can terminate a Las Vegas algorithm after a certain number of queries, turning it into a Monte Carlo algorithm. 
By Markov's inequality, a Las Vegas algorithm with complexity $L$ can be turned into a Monte Carlo algorithm with error parameter $\eps$ and complexity $O(L/\eps)$.
On the other hand, there exist functions whose Las Vegas complexity is much larger than their Monte Carlo complexity.
Ref.~\cite{belovs:separations} features a quadratic separation for a total Boolean function.
For partial functions, the separation can be even larger.

Let us now turn to quantum query complexity.
In the overwhelming majority of cases, the complexity under consideration is Monte Carlo:  
the number of queries is fixed, and the algorithm can output an incorrect output with small probability.

Zero-error quantum algorithms have also been defined and studied~\cite{beals:pol, buhrman:boundsForSmallError, benDavid:quantumDistinguishingComplexity}.
A zero-error quantum algorithm is not allowed to give an incorrect output, but it can output '?' with probability at most $1/2$.
The answer '?' means that the algorithm has not figured out what the answer is.
This model is indeed a quantum counterpart of one way of defining randomised Las Vegas complexity.
However, it lacks the nice features of the randomised Las Vegas complexity outlined above.
The definition depends on the value with which '?' can be outputted, and it also does not compose nicely~\cite{buhrman:quantumZeroErrorCannotBeComposed}.

Another related notion is variable-time model introduced by Ambainis.
In this case, a quantum subroutine can run for an unpredicted number of steps, and the average running time is the corresponding quadratic mean $\sqrt{\sum_t p_t t^2}$, where $p_t$ is the probability the subroutine runs for $t$ steps.
Ambainis showed how to perform search~\cite{ambainis:searchVariableTimes} and amplitude amplification~\cite{ambainis:amplificationVariableTimes} on such subroutines.
A very recent paper by Jeffery~\cite{jeffery:subroutineComposition} also considers quantum walks with such subroutines.
Up to our knowledge, this notion has not been studied as a complexity measure \emph{per se}.
Also, these results are mostly concerned with time complexity, while we study query complexity in this paper.

\subsection{Adversary Bound}
\label{sec:introducitonAdversary}
The quantum adversary bound was first developed as a powerful tool for proving quantum query lower bounds, but it has been later extended to include upper bounds as well.

The adversary bound originates from the hybrid method by Bennett, Bernstein, Brassard, and Vazirani~\cite{bennett:strengths}, which was further refined by Ambainis in the first version of the adversary bound~\cite{ambainis:adv}.
Due to its attractive combinatorial formulation, it fostered a large number of applications~\cite{Durr:quantumGraph, berzina:graphProblems, buhrman:productVerification, dorn:algebraicProperties} to name just a few.

The bound was strengthened by H\o yer, Lee, and \v Spalek~\cite{hoyer:advNegative}.
Using the semidefinite formulation of the adversary bound by Barnum, Saks, and Szegedy~\cite{barnum:advSpectral}, they showed that the same expression still yields a lower bound if one replaces non-negative entries by arbitrary real numbers.  This \emph{negative-weighted} version of the bound is strictly more powerful than the positive-weighted one, but it is also harder to apply.
In a series of papers~\cite{reichardt:formulae, reichardt:spanPrograms, reichardt:advTight}, Reichardt \etal surprisingly proved that the negative-weighted version of the bound is \emph{tight}:  The dual formulation of the bound (which is equal to the primal formulation due to strong duality) can be transformed into a quantum query algorithm with the same complexity up to a constant factor.

The negative-weighted adversary bound has been used to prove lower bounds~\cite{belovs:kSumLower, belovs:onThePower, belovs:setEquality}, but more frequently to prove upper bounds, in particular using the learning graph approach~\cite{belovs:learning}.  For instance, the adversary bound (sometimes in the equivalent form of span programs) was used to construct quantum algorithms for formula evaluation~\cite{reichardt:formulae,  reichardt:unbalancedFormulas, zhan:treesWithHiddenStructure}, finding subgraphs~\cite{lee:learningTriangle, belovs:learningClaws, legall:constSizedHypergraphs}, $k$-distinctness problem~\cite{belovs:learningKDist}, and in learning and property testing~\cite{belovs:learningSymmetricJuntas, belovs:monotonicityQuantum}.

The next steps came when the adversary bound was extended to state generation by Ambainis, Magnin, R\"otteler, and Roland~\cite{ambainis:symmetryAssisted}; and state conversion by Lee, Mittal, Reichardt, \v Spalek, and Szegedy~\cite{lee:stateConversion}.
Belovs~\cite{belovs:variations} extended the bound for various types of input oracles, including the case when the input oracle can be an arbitrary unitary.
These generalisations came with a twist, as the bound became \emph{semi-tight}: a lower bound for the \emph{exact} version of the problem and an upper bound for the \emph{approximate} version of the bound.

Let us briefly touch on techniques used in the above papers.
Ambainis~\cite{ambainis:adv} and H\o yer \etal~\cite{hoyer:advNegative} only proved lower bounds, which they did considering a so-called progress function of the algorithm.
The upper bounds by Reichardt~\cite{reichardt:formulae, reichardt:spanPrograms, reichardt:advTight} used a rather complicated quantum walk, which was inspired by previous work on evaluating NAND-trees~\cite{farhi:nandTree, ambainis:formulaeEvaluation}.
The (discrete) quantum walk comprises two reflections, one simple and input-dependent, and the other one complicated and input-independent.
The analysis of the algorithm required technically involved spectral analysis.

The paper by Lee \etal~\cite{lee:stateConversion} featured many important technical innovations.
First, the problem was generalised to \emph{state conversion}, where the task of the algorithm is to transform one vector $\xi_x$ into another $\tau_x$ on every input $x$ in the domain $D$.
This turned out to be a very fruitful approach, as the algorithm can be broken into smaller steps, which can be then analysed independently.
Second, a very simple proof of the lower bound was presented, which worked by a direct conversion of the algorithm into the bound.
This essentially established the adversary bound as a semi-definite relaxation of the algorithm.
Third, the bound was formulated as an instance of filtered $\gamma_2$-norm, which is a generalisation of $\gamma_2$-norm used previously in other context, see \rf{sec:onegamma} for more detail.
Finally, the proof of the upper bound was significantly simplified by introducing easy and powerful Effective Spectral Gap Lemma to analyse the resulting quantum walk.
The lemma can be also used independently~\cite{belovs:electicityQuantumWalks, belovs:mergedWalk3Dist}.

Lee~\etal~\cite{lee:stateConversion} assumed the standard input oracle that encodes a string $x\in[q]^n$ for some alphabet $[q]$.
The problem of choice by Belovs~\cite{belovs:variations} was still state conversion $\xi_x\mapsto \tau_x$ but this time with \emph{general input oracles} $O_x$, which is just an arbitrary unitary transformation.
This removed the oracle-specific details from the proof, thus making it more transparent.
(Barnum~\cite{barnum:unitaryInputOracle} already considered the problem of function evaluation with unitary input oracles, but that paper went unnoticed at the time.)
The bound was formulated as an instance of relative $\gamma_2$-norm, which further generalises filtered $\gamma_2$-norm, and which, in our opinion, is more natural than the latter.
Belovs used the adversary bound for this problem to construct various adversaries for function and relation evaluation.

Finally, let us note that all the versions considered above are that of the so-called 
\emph{additive} adversary bound.
We leave out of consideration the multiplicative version by \v Spalek~\cite{spalek:multiplicative} based on the earlier work of Ambainis~\cite{ambainis:kFoldedSearch}.


\subsection{Our Results and Techniques}
\label{sec:ourResults}

\paragraph{Las Vegas Complexity.}
We propose a different definition of quantum Las Vegas query complexity, which is very natural and more quantum in spirit than the previous notion of zero-error quantum algorithm.
We define it as the total sum of the squared norms of all the states processed by the input oracle during the execution of the algorithm.
Since the square of the norm means probability in the quantum world, this quantity can be interpreted as the expected number of queries performed by the algorithm on input $x$.

Our variant of quantum Las Vegas complexity possesses all the nice properties mentioned above.
It does not feature any artificial constants.  
It composes nicely as we will show in Sections~\ref{sec:LasVegasProperties} and~\ref{sec:compositionRevisited}.
Finally, as we will show in \rf{sec:mainupper}, every quantum Las Vegas algorithm with complexity $L$ can be turned into a Monte Carlo algorithm with error $\eps$ and complexity $O(L/\eps)$.
Since the term `zero-error' is standard for the previous definition, we call our version `Las Vegas'.

Contrary to Monte Carlo complexity, Las Vegas complexity is input-dependent: different inputs can have different complexity.
Thus, we can study not only the worst-case complexity, but also track complexity on each input.
We capture this by introducing complexity profile, which is the vector recording the complexity of the algorithm on all inputs.

Very recently, and independently of our paper, Jeffery~\cite{jeffery:subroutineComposition} came up with essentially the same notion, which was combined with variable-time quantum algorithm to get a composition result.
The results of \rf{sec:compositionRevisited} can be seen as a query analogue of the latter composition result.

\paragraph{New Simplified Upper Bound Construction.}
As mentioned in \rf{sec:introducitonAdversary}, 
our paper continues the line of work relating quantum query algorithms and the adversary bound.
Following Lee~\etal~\cite{lee:stateConversion}, the relation between the two can be depicted as in \rf{fig:correspondence}.
\myfigure{\label{fig:correspondence}}
{Correspondence between quantum query algorithms and the adversary bound}
{
\negbigskip
\negmedskip
\[
\begin{tabular}{rc@{$\longleftrightarrow$\quad}l}
Computational problem && Adversary optimisation problem\\
Input of the problem && Variable-vector in the problem\\
Query algorithm solving the problem && Feasible solution to the problem\\
Complexity of the algorithm && Objective value of the feasible solution
\end{tabular}
\]
}

After formulating the adversary optimisation problem in Row 1 of \rf{fig:correspondence}, the main issue is to prove the corresponding lower and upper bounds.
A lower bound is a transformation of an algorithm into a feasible solution in Row 3 of \rf{fig:correspondence}, which respects the fourth row of the same diagram.
An upper bound is a transformation in the opposite direction, which turns a feasible solution into an algorithm.

For the lower bound, we follow the same approach that was developed in~\cite{lee:stateConversion} and used in~\cite{belovs:variations}.
The variable-vector is the direct sum of all the queries made to the oracle on the corresponding input.
Therefore, its squared norm lower bounds query complexity, as the state processed on one query has norm at most 1.

The crucial novel ingredient in our paper is a new construction of the upper bound.
The idea behind it is as follows.
What we would like to do is to reverse the above process and to give the variable-vector from the feasible solution as a query to the input oracle.
At the first sight, it is not clear how to achieve this.
Indeed, the algorithm does not know the input, hence, does not know which vector to give.
And even if it knew, the latter vector generally has norm much larger than 1, making it impossible to use it as a query.

We have found a very simple way around these complications.
Assume we add a small ``catalyst'' to the state of algorithm, which is just a scaled down variable-vector from the feasible solution.
We process the catalyst by the input oracle, as we wanted, and use the result to change a tiny part of the state in the required direction.
The constraints of the adversary optimisation problem ensure that the latter step can be implemented by an input-independent unitary.
What is remarkable, however, is that we get the catalyst back after this unitary!
So we can use it again, and again, until we perform the required transformation on all of the state.
Thus, it suffices for the algorithm to ``guess'' the catalyst just once to perform the transformation described above.

The ``guessing'' ability is folklore in quantum algorithms.
What is meant here is that if the catalysis is small, the distance between the original state and the state with the catalyst is also small, and it gets preserved during the execution of the algorithm.
Therefore, we end up close to the target state even if we started without the catalyst.
The smaller the catalyst, the larger the number of queries needed, but the smaller the error induced by guessing it.

Let us compare our algorithm with two previous approaches.
They are the aforementioned quantum-walk-based algorithm by Lee \etal~\cite{lee:stateConversion} and an adiabatic algorithm by Brandeho and Roland~\cite{brandeho:adiabaticAdversary}.
Both of them use the guessing ability to extend the initial state with a small state incorporating the feasible solution.
After that the algorithm uses a quantum walk or an adiabatic transformation, respectively.

What we demonstrated is that the same effect can be obtained by a very simple unitary transformation.
This substantially simplifies the analysis and makes the algorithm more transparent.
In particular, we see what queries are being made by the algorithm: it repeatedly calls the input oracle on the scaled down variable-vector from the feasible solution.

This allows us to make various improvements.
First, we see that the Las Vegas complexity of the algorithm is exactly the objective value of the optimisation problem. 
Second, we can easily incorporate multiple input oracles.
Third, the upper bound works assuming \emph{unidirectional} access to the input oracle, while the previous algorithm in~\cite{belovs:variations} required bidirectional access (the algorithm can query both the input oracle $O_x$ and its inverse $O^*_x$).
Fourth, we do not even need the input oracle to be unitary.
Finally, we get a slightly better dependence of the number of queries (in the traditional sense) on the error parameter $\eps$, which is now tight up to constant factors.
Let us further discuss these improvements.

\paragraph{Relation to Las Vegas Complexity.}
With the new upper bound, it becomes easy to calculate the Las Vegas complexity of the algorithm, which leads to the main result of this paper:
\begin{center}
\emph{The quantum adversary bound is equal to the Las Vegas complexity of state conversion.}
\end{center}
We note that this is a threefold tighter connection between the adversary bound and the usual (Monte Carlo) query complexity.
First, the bound is \emph{tight}, while connection to Monte Carlo complexity is semi-tight.
Second, the bound is \emph{exactly} equal to Las Vegas complexity, and not just up to a constant factor.
Third, the bound holds for \emph{all inputs simultaneously}, and not just worst-case.

This result automatically carries over to all special cases of state conversion, including state generation and function evaluation.

\paragraph{Multiple Input Oracles.}
Considering multiple input oracles is often useful.
For instance, even the standard input oracle $O_x\colon \ket|i>\ket|b>\mapsto \ket|i>\ket|b \oplus x_i>$ is a direct sum of multiple input oracles $O_{x_i}\colon \ket|b>\mapsto\ket|b\oplus x_i>$, which encode individual symbols of the input string $x$.

These settings were investigated previously, most notably in the context of compositional results.
Reichardt and \v Spalek~\cite{reichardt:formulae} considered span programs with costs, where costs were assigned to individual symbols, and which were meant to capture the complexity of the corresponding subproblems, and Ref.~\cite{reichardt:spanPrograms} similarly consider the adversary bound with costs.

Multiple oracles are also necessary in the study of trade-offs between different input resources.
Kimmel, Lin, and Lin~\cite{kimmel:oraclesWithCosts} used an adversary-based approach to show a trade-off between two input oracles.  Again, the adversary featured costs.
Belovs and Rosmanis~\cite{belovs:counting} used a similar approach with general input oracles.
Actually, the whole notion of Las Vegas complexity, including the multiple oracle case, is greatly inspired by the latter paper.

In the case of Las Vegas complexity, dealing with several input oracles is easy.
We can use the same definition (the sum of the squared norms of the states processed by the oracle) for each oracle independently.
The complexity profile becomes a matrix, where, for each input, the complexity of each of the input oracles is listed.

In the adversary bound, the variable-vector is similarly broken down into parts corresponding to different input oracles, and all the results carry over with minimal changes.
Having a vector of complexities of all the input oracles, it is easy to get compositional results as well as trade-offs.

\paragraph{Unidirectionality.}
The upper bound in~\cite{belovs:variations} used quantum walk, which imposed bidirectional access to the input oracle to implement the required reflection.
In the new upper bound, there are no reflections, hence, there is no need for this assumption.

Allowing bidirectional access has a lot of rationale.
First, oracles are usually thought as quantum subroutines, and each quantum subroutine can be easily inverted.
Also, many basic quantum algorithms like Grover's search and amplitude amplification often require bidirectional access to work.

On the other hand, unidirectional access also naturally comes up in some cases.
For instance, if we send the state to some other party to apply the input oracle,
we may trust the recipient to perform the required query, but we might not be able to ask them to perform it in reverse.

Finally, the assumption that we have unidirectional access to the input oracle simplifies the upper and the lower bounds.
The bidirectional case can be obtained as a special case, 
see \rf{sec:bidirectionality}.

\paragraph{General Input Oracles.}
As there is no bidirectionality assumption, we can replace the unitary oracle assumed in~\cite{belovs:variations} by an arbitrary linear transformation.
It turns out that many of results hold still hold even under such assumptions.
It seems, however, that contraction oracles, which are linear transformations of norm not exceeding 1, might be a good choice to consider.

Contraction oracles seem to contradict the unitarity condition usually imposed on quantum algorithms.
Nonetheless, such oracles naturally come up in practise.
For example, the input oracle can perform some measurement and continue only if the outcome is positive.
Similar settings appear in interaction-free measurements by Elitzur and Vaidman~\cite{elitzur:interactionFree} and subsequent bomb query algorithm by Lin and Lin~\cite{lin:bombTester}, measure-many quantum finite automata~\cite{kondacs:measureManyQuantumAutomata}, and faulty oracles~\cite{regev:faultyOracle}.

\paragraph{Subspace Conversion Problem.}
Finally, we define and study the subspace conversion problem, which is in between state conversion $\xi_x\mapsto\tau_x$, which we assume for the action of the algorithm, and unitary (or contraction) $O_x$, which we use for the input oracle.

In this problem, the task is to implement a linear transformation $T_x\colon \cK_x\to \cK$ defined on a linear subspace $\cK_x$ of the workspace $\cK$.
If $\cK_x$ is one-dimensional, this is state conversion; 
if $\cK_x = \cK$, this is unitary (or contraction).

We introduce a complexity notion for this problem, which is the largest Las Vegas complexity of the algorithm achieved when executed on a unit vector in $\cK_x$.
The definition turns out to be natural for composition, and it is still exactly characterised by the corresponding version of the adversary bound.

\paragraph{Unitary Permutation Inversion.}
Finally, we use this occasion to demonstrate a separation between unidirectional and bidirectional access to the input oracle.
We consider the unitary permutation inversion problem, where the oracle is a unitary that implements a permutation $\k i \mapsto \ket|\pi(i)>$, and the task is to find $\pi^{-1}(1)$.  
We prove an $\Omega(\sqrt n)$ lower bound, where $n$ is the size of the domain of $\pi$, whereas the problem is trivially solvable with 1 query to the inverse oracle.
Up to our knowledge, these types of questions have not been addressed before.

\mycutecommand\Ic{I^\circ}
\mycutecommand\Ib{I^\bullet}
\mycutecommand\State{\cS}
\mycutecommand\Query{\cQ}

\subsection{Overview of the Paper}
In this subsection, we give a very brief overview of the paper, highlighting the most important points.

The main part of the paper starts with \rf{sec:definitions}.
In \rf{sec:algorithm}, we define a quantum query algorithm as a sequence of linear transformations
\begin{equation}
\label{eqn:preAlgorithm}
U_T\, \tO\, U_{T-1}\,\tO\,\cdots U_{1}\, \tO\, U_0,
\end{equation}
where $U_i$ are some unitaries. 
The operator $\tO = (O\otimes \Ib)\oplus \Ic$ is a query, where $O$ is the input oracle, and $\Ib$ and $\Ic$ are some identity transformations.
Thus, the algorithm implements a transformation $O\mapsto \cA(O)$: from the input oracle to the linear operator~\rf{eqn:preAlgorithm}.
In \rf{sec:conditions}, we describe problems solved by the algorithm.
We first give a general definition, capable of describing a wide range of problems, but for the purposes of this paper, the most important problem is state conversion.
Given a family of input oracles $O_x$ and pairs $\xi_x\mapsto \tau_x$, where $x$ ranges overs some set $D$, the task is to develop an algorithm $\cA$ such that $\cA(O_x)$ maps $\xi_x$ into $\tau_x$ for all $x\in D$.

The remaining part of the paper follows a similar division.
Sections~\ref{sec:LasVegas} and~\ref{sec:LasVegasProperties} are devoted to the study of Las Vegas complexity of algorithms without connection to any particular problem.
In \rf{sec:advStateConversion}, we consider the state conversion problem, and in \rf{sec:subspaceConversion}, subspace conversion.
In particular, the adversary bound makes its first appearance in \rf{sec:advStateConversion} as it is tied to a particular problem being solved.
Other problems can be studied as well, for instance, in \rf{sec:booleanFunction}, we consider the problem of Boolean function evaluation, and Ref.~\cite{belovs:variations} considers a wide range of other problems, which we leave outside the confines of this paper.
\rf{sec:onegamma} is an intermission, and Sections~\ref{sec:bidirectionality} and~\ref{sec:permutationInversion} contain complementary results.
Let us describe these sections in more detail.

In \rf{sec:LasVegas}, we define Las Vegas complexity a quantum query algorithm $\cA$.
Let $\Query_t(\cA,O)\xi$ be the state processed by the input oracle (the $O\otimes \Ib$ part of the query operator $\tO$) on the $t$-th query when executed on the input oracle $O$ and the initial state $\xi$.
We define the Las Vegas complexity $L(\cA,O,\xi)$ as the sum of $\|\Query_t(\cA,O)\xi\|^2$ over all $t$.
Note that it depends both on the input oracle $O$ and the initial state $\xi$.
In \rf{sec:multipleOracles}, we define the same notion for multiple input oracles.
In essence, the complexity $L(\cA,O,\xi)$ becomes a tuple which accounts for the total squared norm of the states processed by each of the input oracles.
The difference between the single-oracle and the multiple-oracles variants is mostly cosmetic.
The reader may choose to assume the single-oracle variant throughout the paper.

In \rf{sec:LasVegasProperties}, we study various properties of Las Vegas complexity without relation to any particular task.
We consider various ways algorithms can be composed: inversion, direct sum, tensor product, sequential and functional composition, and show that our definition of Las Vegas complexity encompasses these composition variants naturally.
The results are pretty straightforward, but there is one subtlety involving functional composition, where one algorithm $\cB$ is used as an input oracle for another algorithm $\cA$.
The thing is that the algorithm $\cA$ executes the input oracle as $O\otimes \Ib$, while we assume the algorithm $\cB$ implements $O$.
This means that the complexity of $\cB$ on the state $\psi_t = \Query_t(\cA, \cB)\xi$ is just not defined.
We use an obvious solution to slice $\psi_t$ as $\psi_{t,1}\oplus \cdots \oplus \psi_{t,d}$, where $d$ is the dimension of $\Ib$, and each $\psi_{t,j}$ can be processed by $O$.
Now, the complexity of $\cB$ on each $\psi_{t,i}$ is well-defined, and we can define the total complexity as their sum.
There are many different possible slicing, as they depend on the choice of an orthonormal basis in the space of $\Ib$.
We show that the total complexity is independent of the choice of slicing.

In \rf{sec:onegamma}, we describe our modification to the relative $\gamma_2$-norm from~\cite{belovs:variations}, as different settings require different versions of the bound.
One thing we should account for is unidirectionality.
We also have to convert to multi-objective version of the bound, as we are interested in the full complexity profile of the algorithm.
The variant with multiple input oracles requires yet another modification.
We formulate the dual versions, which can be used to prove lower bounds on worst-case complexity.

\rf{sec:advStateConversion} is the main part of the paper.
In it, we study Las Vegas complexity of state conversion, and show how it can be characterised by an instance of (unidirectional) relative $\gamma_2$-bound: the adversary optimisation problem.
This section is designed to be self-contained with minimal dependency on the previous sections.
We first prove a lower bound for the exact version of the problem in \rf{sec:mainlower}, and then an upper bound for the approximate version in \rf{sec:mainupper}.
The corresponding ideas were already explained in \rf{sec:ourResults}.
The algorithm in \rf{sec:mainupper} transforms
\[
\xi^+_x = \xi_x \oplus \frac1{\sqrt T} v_x
\quad\longmapsto\quad 
\tau^+_x = \tau_x \oplus \frac1{\sqrt T} v_x,
\]
where $\xi_x\mapsto\tau_x$ is the required state conversion problem, $v_x$ is a feasible solution to the adversary optimisation problem, and $T$ is an arbitrary positive integer.
The Las Vegas complexity of the algorithm on input $x$ is $\|v_x\|^2$ independently from the value of $T$.
(The total number of queries does depend on $T$, though).
As $T$ increases, we can get arbitrarily close to the required transformation, while Las Vegas complexity stays the same.
In \rf{sec:equality}, we improve on this result.
We show how to perform transformation $\xi_x\mapsto \tau_x$ \emph{exactly}, while now we can get Las Vegas complexity arbitrarily close to $\|v_x\|^2$.

\rf{sec:subspaceConversion} deals with the question of what happens if some of the input oracles $O_x$ in the state conversion problem are equal.
If $O_x = O_y$, then we get exactly the same action of the algorithm on the inputs $x,y\in D$.
The motive of \rf{sec:linearConsistency} is linear consistency of the feasible solution to the adversary bound for such pairs of inputs.
We show that we can assume consistency without any loss.
Moreover, this brings us to the formulation of the subspace conversion problem, and the corresponding adversary bound in \rf{sec:subspaceConversionProblem}.
In \rf{sec:compositionRevisited}, we revisit the functional composition property from \rf{sec:LasVegasProperties} and show that we can upper bound the complexity of the composed algorithm as the product of complexities of the constituents under the assumption that the algorithm follows the specification of the subspace-converting subroutine.

In \rf{sec:bidirectionality}, we prove the relation between unidirectional and bidirectional versions of the bound.
In particular, we get back the bidirectional results from~\cite{belovs:variations}.
In \rf{sec:permutationInversion}, we prove a separation between unidirectional and bidirectional input oracle for the unitary permutation inversion problem.
Finally, in \rf{sec:discussion}, we make some final comments.

Ref.~\cite{yolcu:adversary} contains an alternative exposition of some of the results in \rf{sec:advStateConversion}, and some additional results on more general control problems.

\section{Preliminaries}
If not said otherwise, a \emph{vector space} is a finite-dimensional complex inner product space.  They are denoted by calligraphic letters.  We assume that each vector space has a fixed orthonormal basis, and we often identify an operator with the corresponding matrix.
The inner product is denoted by $\ip<\cdot,\cdot>$.
$A^*$ stands for the adjoint linear operator, and $A\elem[i,j]$ for the $(i,j)$th entry of the matrix $A$.
$A\circ B$ stands for the Hadamard (entry-wise) product of matrices.
$I_{\cX}$ stands for the identity operator in $\cX$.
All projectors are orthogonal projectors.
For vectors $u,v\in \bR^n$, we write $u\le v$ if $u\elem[i]\le v\elem[i]$ for all $i\in [n]$.
We use $A\succ 0$ and $A\succeq 0$ to denote positive definite and semi-definite matrices, respectively.
We use the ket-notation to emphasise that a vector is a state of a quantum register, or to denote the elements of the computational basis.

We also need the following generalisation of the well-known parallelogram identity.
We were not able to find its statement in the existing literature.

\begin{thm}[Generalised Parallelogram Identity]
\label{thm:parallelogram}
Let $v_1,\dots,v_d\in \bC^n$, and
\begin{equation}
\label{eqn:unitaryU}
U = \begin{pmatrix}
\alpha_{1,1} & \alpha_{1,2} & \dots& \alpha_{1,d}\\
\alpha_{2,1} & \alpha_{2,2} & \dots& \alpha_{2,d}\\
\vdots & \vdots &\ddots & \vdots\\
\alpha_{d,1} & \alpha_{d,2} & \dots &\alpha_{d,d}
\end{pmatrix}
\end{equation}
be a unitary matrix.  Then,
\[
\|v_1\|^2 + \|v_2\|^2 + \cdots + \|v_d\|^2 = \sum_{j=1}^d 
\normA| \alpha_{1,j} v_1 + \alpha_{2,j}v_2 + \cdots + \alpha_{d,j} v_d |^2 .
\]
\end{thm}

\begin{proof}
Let $V$ be the $n\times d$ matrix with $v_j$ as the columns.
The above identity is equivalent to
\[
\normA| V |_{\mathrm F}^2 = \normA| VU |_{\mathrm F}^2,
\]
where $\norm|\cdot|_{\mathrm F}$ stands for the Frobenius norm.
The equality follows from the fact that unitaries preserve the Frobenius norm.
\end{proof}

The parallelogram identity is the special case of \rf{thm:parallelogram} with $U = H$, the Hadamard matrix.

\section{Quantum-Algorithmic Definitions}
\label{sec:definitions}
In this section, we give the main definition of a quantum algorithm solving a computational problem.
Let us very briefly recall the textbook definition of a quantum query algorithm.
A standard reference is a survey by Buhrman and de Wolf~\cite{buhrman:querySurvey}.
(Note, however, that it only deals with Boolean functions.  See also~\cite{cleve:quantumComplexityTheory}.)
The task is evaluation of a function $f\colon D\to [\ell]$ with domain $D\subseteq [q]^n$.  
The algorithm can perform arbitrary unitary transformations, as well as access the input string $x=(x_1,\dots,x_n)\in D$ via the \emph{standard input oracle}:
\[
O_x\colon \ket|i>\ket|b>\mapsto \ket|i>\ket|b \oplus x_i>,
\]
where $\oplus$ is the bit-wise XOR operation (one can also use modular addition).
The algorithm is said to compute the function $f$ if, for all $x\in D$, measuring the output register of the final state of the algorithm gives $f(x)$ with high probability.
The unitary operations are free, and each execution of $O_x$ costs one query.
The goal is to minimise the number of queries.

We separate the algorithm itself from the input and output conditions.
The algorithm becomes a map from operators on the input register ($O_x$) to operators on the space of the algorithm (the transformation performed by the algorithm).
The input condition is the input oracle $O_x$ given on a particular input $x\in D$, and the output condition is the set of admissible transformations performed by the algorithm.
Actually, we choose to treat input and output conditions similarly as sets of admissible transformations, which allows algorithms to be used as input oracles for other algorithms.
We keep the set of input labels $D$, but it need not be considered as the domain of a function any longer.

We describe the algorithmic part in \rf{sec:algorithm}, and the input/output conditions in \rf{sec:conditions}.

\subsection{Quantum Query Algorithm}
\label{sec:algorithm}

The overall form of a quantum query algorithm is similar to the textbook version.
When we use the term `algorithm' later in the paper, we mean a quantum query algorithm of the following form.

\begin{defn}[Quantum Query Algorithm]
\label{defn:algorithm}
Let $\cM$ and $\cH$ be vectors spaces.
A \emph{quantum query algorithm} in $\cH$ with an oracle in $\cM$ is a function 
which maps linear operators $O\colon \cM\to\cM$ into linear operators $\cA(O)\colon \cH\to\cH$, and which has the following form:
\begin{equation}
\label{eqn:algorithm}
\cA(O) = U_T\, \tO\, U_{T-1}\,\tO\,\cdots  U_{1}\, \tO\, U_0.
\end{equation}
Here, each $U_i\colon \cH\to\cH$ is a unitary that does not depend on $O$, and $\tO$ is some ``embedding'' of $O$ into $\cH$ of the form $\tO = (O\otimes \Ib)\oplus \Ic$, where $\Ib$ and $\Ic$ are identity transformations of some size.
\end{defn}

The operator $O$ is called the \emph{input oracle}, and each execution of $\tO$ is called a \emph{query}.
To make the definition simpler, we have chosen to have one fixed embedding $\tO$.
It is often more convenient to allow different embeddings at different queries.
The two definitions are equivalent, see \rf{sec:slicing}.

The spaces $\cM$ and $\cH$ are called the \emph{input} and the \emph{work} spaces of the algorithm, respectively.
It is sometimes useful to consider also the \emph{output} subspace $\cK\subseteq \cH$ of the algorithm.
Conceptually, it contains the ``interesting'' part of $\cA(O)$, while its orthogonal complement in $\cH$ is the ``scratch space'' of the algorithm.
If not specified, we may always take $\cK=\cH$.

If $\cA(O)\xi=\tau$ for some $\xi,\tau\in\cH$ we say that the algorithm $\cA$ performs transformation $\xi\mapsto\tau$ on the oracle $O$.
We call $\xi$ the \emph{initial} and $\tau$ the \emph{terminal} state of the algorithm\footnote{
The letters $\xi$ and $\tau$ stand for \textgreek{ξεκίνημα} and \textgreek{τέλος}, respectively.
}.

Let us mention the main differences with the textbook definition.
The first difference is that we allow arbitrary input oracles $O$. 
The second difference is the ability to ``skip'' query: to apply $I^{\circ}$ on some part of the workspace.
Alternatively, in the language of circuits, we may apply a controlled version of $O$, not just $O$.
Textbook quantum query algorithms do skip queries, but it is usually done implicitly by setting up a state that does not change by any input oracle, e.g., a uniform superposition on the second register.
Since we allow arbitrary unitaries as oracles, this option is out of stock for us, and we have to skip queries explicitly.
Interestingly, this is exactly this feature that allows us to define quantum Las Vegas complexity.

Let us also emphasise the differences with the definition of a quantum query algorithm in~\cite{belovs:variations}.
The first one is what we call directionality.
The algorithm in~\cite{belovs:variations} is \emph{bidirectional}: it allows execution of both $\tO$ as well as its inverse $\tO^{-1}$.
The algorithm in \rf{defn:algorithm} is \emph{unidirectional}: it only allows execution of $\tO$.
For the standard input oracle, the difference is irrelevant since each standard input oracle is its own inverse (or can be easily constructed from it).
This is not true for arbitrary unitaries.
Note that unidirectionality is without any loss of generality, as it is possible to simulate bidirectional access to a unitary $O$ with unidirectional access to $O\oplus O^*$, see \rf{sec:bidirectionality}.

The second difference is that, while we still require that all $U_i$ are unitary, there is no more need to require the input oracle $O$ to be unitary.
We allow $O$ to be an arbitrary linear operator.
However, a more interesting choice is to consider contractions as input oracles.
Note that if $O$ is a unitary or a contraction, then $\cA(O)$ is also a unitary or a contraction, respectively.

\subsection{Input and Output Conditions}
\label{sec:conditions}
Here, we give a general take on input and output conditions imposed on a quantum algorithm, as well as define all types of conditions we consider in this paper. 
In Sections~\ref{sec:advStateConversion} and~\ref{sec:subspaceConversion}, we redefine the specific conditions under consideration.

The simplest way to impose requirements on an algorithm $\cA$ from \rf{defn:algorithm} is to specify its outputs $\cA(O_x)\colon \cH\to\cH$ on fixed inputs $O_x\colon \cM\to\cM$ as $x$ ranges over some set $D$.
There is nothing fundamentally wrong with this approach, except that it may be too specific.
For instance, the textbook definition has a very specific input oracle $O_x$, but a very vague output condition.
To capture this, for each $x\in D$, we define not one, but a collection $\cE_x$ of admissible linear transformations.
This gives the following very general definition.

\begin{defn}[Computational Problem]
A computational problem is given by a set of labels $D$, where, for each $x\in D$, a set of admissible inputs $\cO_x$ and a set of admissible outputs $\cE_x$ are specified.
A quantum algorithm $\cA$ \emph{solves} the problem if, for each $x\in D$ and each $O\in \cO_x$, it holds that $\cA(O)\in \cE_x$.
\end{defn}

We treat input and output uniformly, therefore, we use a term \emph{admissible set} for both sets of admissible inputs and outputs. 
We define different types of admissible sets, which are depicted in \rf{fig:models}.
We do this in terms of $\cE_x$, the output space $\cK$, and the workspace $\cH$.
The definitions for input conditions are similar with $\cE_x$ replaced by $\cO_x$, and $\cK$ and $\cH$ by $\cM$.
We say that we have a problem of type$_1$ with input oracles of type$_2$, if all $\cE_x$ are of type$_1$ and all $\cO_x$ are of type$_2$.

\myfigure{\label{fig:models}}%
{Various types of input/output conditions.  
The ones at the top are more general in the sense that the lower ones are special cases thereof as indicated by arrows.
The ones to the right are more restrictive in the sense that they impose more restrictions on the set of admissible operators.
}%
{
\negbigskip
\[
\begin{tikzpicture}
\node at (0,0) (SbC) {Subspace Conversion} ;
\node at (-5,-1) (StC) {State Conversion} ;
\node at (5,-1.1) (General) {\txt{General Input Oracle\\(Unitary/Contraction)}} ;
\node at (-5,-2.25) (StG) {State Generation} ;
\node at (-5,-3.5) (Function) {Function Evaluation} ;
\graph[use existing nodes]{
SbC<-StC<-StG<-Function;
SbC<-General;
};
\end{tikzpicture}
\]
}

Most types of admissible sets considered in this paper are special cases of the following type.
\begin{defn}[Subspace Conversion]
\label{defn:subspaceConversion}
An admissible set $\cE_x$ is \emph{subspace conversion} if there exists a linear transformation $S_x\colon \cK_x\to\cK$ defined on some linear subspace $\cK_x\subseteq \cK$ such that $\cE_x$ consists of all extensions of $S_x$ to a linear operator on $\cH$.

There are two main cases.
In the \emph{isometric} case, we only allow unitaries in $\cE_x$.  Of course, this makes sense only if $S_x$ is an isometry itself.
In the general \emph{non-isometric} case, we assume that $S_x$ are contractions, and require the operators in $\cE_x$ to be contractions as well.
\end{defn}

Therefore, subspace conversion is specified by its action on the output space $\cK$, but accommodates any workspace $\cH$ as long as it is a superspace of $\cK$.
The vectors in $\cK\setminus\cK_x$ are interpreted as ones where the action of the algorithm is not defined.  Note that it is possible that $A\in\cE_x$ maps such vectors outside of $\cK$.

There are two main special cases of subspace conversion, which are more important than the general case itself.
The first one is when $\cK_x = \cK$.
In this case, $S_x$ gives a linear map from $\cK$ to itself.
The algorithm can still use a larger workspace, but it is completely inaccessible from outside, therefore, it makes sense to identify $\cE_x$ with $S_x$.
This is our default type of input condition, which we call \emph{general input oracle}.
Alternatively, we call it unitary, contraction, or linear input oracle in dependence on the type of $S_x$.
For the output condition, we call it \emph{unitary} or \emph{contraction implementation}.

The second important special case is when $\cK_x$ is one-dimensional.
We call it \emph{state conversion}, and denote by $\xi_x\mapsto \tau_x$, meaning that $A\xi_x = \tau_x$ for all $A\in\cE_x$.
This is our default type of output condition.

There are important special cases of state conversion as well.
\emph{State generation} is state conversion when all the initial states $\xi_x$ are equal to some predefined state $\ket|0>$.
The most widely used version is \emph{function evaluation}, which is state generation when $\tau_x$ is an element of the computational basis $\ket|f(x)>$ for some function $f\colon D\to K$.

It is also possible to define approximate and non-coherent versions of above conditions.
In the $\eps$-approximate version, we take the $\eps$-neighbourhood of $\cE_x$.
For instance, an algorithm $\cA$ solves an $\eps$-approximate version of state conversion $\xi_x \mapsto \tau_x$ if, for all $x\in D$ and all $O\in\cO_x$, we have $\| \cA(O)\xi_x - \tau_x\|\le \eps$.
We say that $\cA$ solves the non-coherent version of the problem, if $\cA(O)\xi_x = \tau_x\otimes\zeta$ for some junk state $\zeta$ that may depend on $x$ and $O$.
Finally, we can consider $\eps$-approximate non-coherent version as well, where we require that $\| \cA(O)\xi_x - \tau_x\otimes\zeta\|\le \eps$.

Function evaluation is usually considered in the approximate non-coherent case, as it is required that measuring the output register of the final state gives $f(x)$ with bounded error.
However, for bidirectional oracles, coherent and non-coherent versions differ at most by a factor of 2 in complexity.  Indeed, it is possible to evaluate the function non-coherently, copy the final output into a new register, and run the program in reverse.
For unidirectional oracles, however, this simple trick does not work, as it is impossible to run the program in reverse.
It also does not work for state generation, as it is impossible to copy general quantum state.

\section{Quantum Las Vegas Query Complexity}
\label{sec:LasVegas}

In this section, we define the main notion of this paper: quantum Las Vegas query complexity.
Usually query complexity of the algorithm like in \rf{defn:algorithm} is defined as $T$: the number of invocations of the input oracle.
We will often call it \emph{Monte Carlo} query complexity in this paper.
Contrary to Monte Carlo complexity, Las Vegas complexity is input-dependent, as it depends both on the oracle $O$ and the initial state.

\subsection{Definition}
\label{sec:LasVegasDefinition}

Let $\cA$ be an algorithm as in \rf{defn:algorithm}, and $O\colon \cM\to\cM$ be an input oracle.
We need the following two linear transformations on the workspace $\cH$, which can be seen as partial executions of the algorithm.
For $t\in[T+1]$, let
\begin{equation}
\label{eqn:State}
\State_t (\cA, O) =  U_{t-1}\, \tO\, U_{t-2}\,\tO\,\cdots U_{1}\, \tO\, U_0
\end{equation}
be the transformation that maps the initial state $\xi$ to the state just before the $t$-th application of the input oracle $O$.
In particular, $\State_0 (\cA, O) = U_0$ and $\State_{T+1} (\cA, O) = \cA(O)$.

Recall that the query is of the form $\tO = (O\otimes \Ib)\oplus \Ic$.
Let $\Pi$ denote the projection on the part of the space processed by $O\otimes \Ib$.
The second transformation is
\begin{equation}
\label{eqn:Query}
\Query_t (\cA, O) = \Pi\State_t (\cA, O),
\end{equation}
which maps $\xi$ to the state processed by the input oracle on the $t$-th query.


\begin{defn}
\label{defn:LasVegas}
The \emph{quantum Las Vegas query complexity} of the algorithm $\cA$ on the input oracle $O\colon \cM\to\cM$ and the initial state $\xi\in \cH$ is defined as
\begin{equation}
\label{eqn:LasVegasDefinition}
L(\cA, O, \xi) = \sum_{t=1}^T \normA|\Query_t (\cA, O)\xi|^2.
\end{equation}
\end{defn}

Under usual assumptions of $O$ being unitary and $\|\xi\|=1$, the term $\norm|\Query_t (\cA, O)\xi|^2$ can be interpreted as the probability that the algorithm $\cA$ actually executes the query on the $t$-th step, and not skips it.
Therefore, $L(\cA, O, \xi)$ can be seen as the expected number of queries similarly to the definition of the randomized Las Vegas query complexity.
Las Vegas complexity does not exceed the Monte Carlo complexity $T$, but it can be much smaller.

The definition also encapsulates the case of algorithms with intermediate measurements as we briefly discuss here.
Assume we have a quantum algorithm $\cB$ with intermediate measurements.
The definition is similar to \rf{defn:algorithm} with the difference that the algorithm can perform measurements in the middle, so that the forthcoming unitaries $U_i$ depend on the outcome of the previous measurements.
In particular, the number of queries can also depend on the outcomes of the measurements.
Let $T(\cB, O, \xi)$ be the expected number of queries performed by $\cB$ on oracle $O$ and initial state $\xi$.
Such an algorithm can be turned into a usual algorithm $\cA$ as in \rf{defn:algorithm} by 
deferring the measurements to the end of the algorithm~\cite[Section 4.4]{chuang:quantum}.
It is not hard to see that $T(\cB, O, \xi)\ge L(\cA, O,\xi)$.
Note, however, that in the absence of measurements, the terminal state of $\cA$ differs from the terminal state of $\cB$.
In particular, $\cA$ computes the non-coherent version of a state conversion problem even if the original algorithm $\cB$ computes the coherent version.

\subsection{Multiple Input Oracles}
\label{sec:multipleOracles}

Assume we have $s$ input oracles, $O^{(1)},O^{(2)},\cdots,O^{(s)}$, where $O^{(i)}$ acts on some space $\cM^{(i)}$, and we want to provide the algorithm with access to all of them.
This can be seen as a special case of \rf{defn:algorithm}, where the algorithm has access to the combined oracle
\begin{equation}
\label{eqn:combinedOracle}
O = O^{(1)}\oplus O^{(2)}\oplus\cdots\oplus O^{(s)}
\end{equation}
acting on $\cM = \cM^{(1)}\oplus\cdots\oplus \cM^{(s)}$.
Indeed, it is possible to simulate a query to $O^{(i)}$ using one query to $O$, and it is possible to simulate a query to $O$ using one query to each of $O^{(i)}$.

Now suppose we want to measure complexity of each oracle $O^{(i)}$ individually.
In the case of Las Vegas complexity, this can be handled very naturally.
Decompose
\begin{equation}
\label{eqn:queryDecomposition}
\Query_t(\cA, O)\xi = \Query_t^{(1)}(\cA, O)\xi\;\oplus\; \Query_t^{(2)}(\cA, O)\xi\;\oplus\;\cdots\;\oplus\; \Query_t^{(s)}(\cA, O)\xi,
\end{equation}
where $\Query_t^{(i)}(\cA, O)\xi$ is the state processed by the $i$-th input oracle on the $t$-th query.

\begin{defn}
\label{defn:LasVegasMultipleOracles}
In the above settings, the \emph{Las Vegas complexity of the $i$-th input oracle} is defined as 
\[
L^{(i)}(\cA, O, \xi) = \sum_{t=1}^T \norm|\Query_t^{(i)} (\cA, O)\xi|^2.
\]
The Las Vegas complexity $L(\cA, O, \xi)$ of the algorithm $\cA$ on the composed input oracle $O$ from~\rf{eqn:combinedOracle} is the vector in $\bR^s$ consisting of the individual complexities $L^{(i)}(\cA, O, \xi)$.
\end{defn}

Almost all the results in this paper can be generalised to include this variation of Las Vegas complexity with minimal changes in the proof.
To make this explicit, we introduce the following piece of notation.
Let $v\in \cM\otimes \cW$ for some $\cW$.
We have the following imposed decomposition
\[
v = v^{(1)}\oplus v^{(2)}\oplus \cdots\oplus v^{(s)},
\]
with $v^{(i)}\in \cM^{(i)}\otimes\cW$.
We define
\begin{equation}
\label{eqn:Dnorm}
\Dnorm|v|^2 = \sB[ \normA|v^{(1)}|^2, \normA|v^{(2)}|^2,\dots, \normA|v^{(s)}|^2 ] \in \bR^s.
\end{equation}
This notation is chosen to emphasise similarity to $\|v\|^2$, and we never use $\Dnorm|v|$ alone.
This gives us almost the same definition for Las Vegas complexity as in~\rf{eqn:LasVegasDefinition}:
\begin{equation}
\label{eqn:LasVegasMultipleOracles}
L(\cA, O, \xi) = \sum_{t=1}^T \DnormB|\Query_t (\cA, O)\xi|^2.
\end{equation}

The upcoming sections can be read using one of the two assumptions:
\begin{itemize}
\item There is a single input oracle $O$.  
In this case, definitions from~\rf{sec:LasVegasDefinition} hold, $s=1$ everywhere, and $\Dnorm|v|^2$ stands for $\|v\|^2$.  In particular, Eq.~\rf{eqn:LasVegasDefinition} and~\rf{eqn:LasVegasMultipleOracles} are the same.
\item There are multiple input oracles.  
In this case, we use $O$ as in~\rf{eqn:combinedOracle} to combine them in a single input oracle.
We use \rf{defn:LasVegasMultipleOracles}, and $\Dnorm|v|^2$ is as in \rf{eqn:Dnorm}.
\end{itemize}
Most of the time, there is no difference between the two cases.

Let us list the properties of $\Dnorm|v|^2$ that we will need.
They follow easily from the definition~\rf{eqn:Dnorm}.
%
\begin{subequations}
\begin{align}
&\Dnorm|cv|^2 = |c|^2\cdot \Dnorm|v|^2 \label{eqn:DnormScaling} \\
&\Dnorm|u\oplus v|^2 = \Dnorm|u|^2 + \Dnorm|v|^2 \label{eqn:DnormSum}\\
&\text{If $O$ is a unitary of the form in~\rf{eqn:combinedOracle}, then $\Dnorm|v|^2 = \Dnorm|Ov|^2$.} \label{eqn:DnormUnitary}
\end{align}
\end{subequations}

Finally, the generalised parallelogram identity also holds.  Namely, in assumptions of~\rf{thm:parallelogram}:
\begin{equation}
\label{eqn:DnormParallelogram}
\Dnorm|v_1|^2 + \Dnorm|v_2|^2 + \cdots + \Dnorm|v_d|^2 
= 
\sum_{j=1}^d \DnormA| \alpha_{1,j} v_j + \alpha_{2,j}v_j + \cdots + \alpha_{d,j} v_d |^2 .
\end{equation}

\section{Properties of Las Vegas Complexity}
\label{sec:LasVegasProperties}
Apart from functional composition, which was the main focus of previous work, algorithms can be composed in many different ways, some of which we describe in this section.
Most of them were used before implicitly, and one of our goals was to formulate them in a more explicit way.

We also show that quantum Las Vegas complexity can handle these composition variants naturally.
Most of the results hold for linear input oracles, but we require unitary input oracles for some.

\subsection{Basic Properties}

\begin{prp}[Scaling]
\label{prp:scaling}
For every algorithm $\cA$, oracle $O\colon \cM\to\cM$, and states $\xi,\tau\in\cH$, if $\cA$ transforms $\xi\mapsto\tau$ on $O$, then it also transforms $c\xi\mapsto c\tau$ for all $c\in\bC$ and
\[
L(\cA, O, c\xi) = |c|^2 L(\cA, O, \xi).
\]
\end{prp}

\begin{proof}
This follows from the definition~\rf{eqn:LasVegasMultipleOracles} and~\rf{eqn:DnormScaling}.
\end{proof}

Note that while $\Query_t(\cA, O)$ is linear, it distorts inner products even if $O$ is a unitary.
Hence, there is no general way to relate $L(\cA, O, \xi + \xi')$ to $L(\cA, O, \xi)$ and $L(\cA, O, \xi')$ even for orthogonal $\xi$ and $\xi'$.
However, we have the following result.
\begin{prp}[Parallelogram Identity]
\label{prp:parallelogram}
For every algorithm $\cA$, oracle $O\colon \cM\to\cM$, states $\xi_1,\dots,\xi_d\in\cH$, and unitary $U$ as in~\rf{eqn:unitaryU}, we have
\[
L(\cA, O, \xi_1) + L(\cA, O, \xi_2) + \cdots + L(\cA, O, \xi_d) = 
\sum_{j=1}^d 
L\sA[\cA, O,  \alpha_{1,j} \xi_1 + \alpha_{2,j}\xi_2 + \cdots + \alpha_{d,j} \xi_d ].
\]
\end{prp}

\pfstart
The proof is analogous to \rf{prp:scaling}, but this time we use~\rf{eqn:DnormParallelogram}.
\pfend

\begin{prp}[Inversion]
\label{prp:inversion}
For every algorithm $\cA$ in $\cH$ with oracles in $\cM$, there exists the inverse algorithm $\cA^{-1}$ in the same spaces such that for every unitary input oracle $O\colon \cM\to\cM$, we have $\cA^{-1}(O^*) = \sA[\cA(O)]^{-1}$.
Moreover, if $\cA$ transforms $\xi\mapsto \tau$ on a unitary input oracle $O$, then
\[
L(\cA^{-1}, O^*, \tau) = L(\cA, O, \xi) .
\]
\end{prp}

\pfstart
The algorithm $\cA^{-1}$ is just the inverse of~\rf{eqn:algorithm}:
\[
\cA^{-1}(O) = U_0^*\, \tO\, U_{1}^*\,\tO\,\cdots U^*_{T-1}\, \tO\, U_T^*.
\]
The relation between Las Vegas query complexities follows from the identity
\[
\Query_t (\cA^{-1}, O^*)\tau = (O\otimes \Ib)\Query_{T+1-t} (\cA, O)\xi
\]
and~\rf{eqn:DnormUnitary}.
\pfend

\subsection{Slicing}
\label{sec:slicing}
Let us now describe possible alternatives to the \rf{defn:algorithm} of the quantum query algorithm, and show that they preserve Las Vegas complexity.
In particular, we show that we can replace the ``embedding'' $\tO = (O\otimes \Ib)\oplus \Ic$ with a simpler construction.

\begin{defn}[Sliced Algorithm]
\label{defn:sliced}
We call a quantum algorithm from \rf{defn:algorithm} \emph{sliced} if its query $\tO$ is of the form $\tO = O \oplus \Ic$.
\end{defn}

Clearly, a sliced algorithm is a special case of the general algorithm.
In the other direction, we have the following result.

\begin{prp}[Slicing]
\label{prp:slicing}
Every algorithm $\cA$ can be transformed into a sliced algorithm $\cA'$ such that, for every oracle $O\colon\cM\to\cM$ and initial state $\xi\in \cH$, we have $\cA(O) = \cA'(O)$ and $L(\cA, O, \xi) = L(\cA', O, \xi)$.
\end{prp}

\pfstart
Let $\tO = (O\otimes \Ib)\oplus \Ic$ be the query of the algorithm $\cA$.
We can rewrite
\begin{equation}
\label{eqn:slicingA}
O\otimes \Ib 
= O\oplus O \oplus \cdots \oplus O
= (O\oplus I \oplus \cdots \oplus I)
(I\oplus O \oplus \cdots \oplus I)
\cdots
(I\oplus I \oplus \cdots \oplus O),
\end{equation}
where there are $d = \dim \Ib$ multipliers on the right-hand side.
Conjugating each of them by a unitary, we can implement $\tO$ using $d$ queries to $\tO' = O\oplus {\Ic}'$.
This does not change the action of the algorithm.

Neither does this change its Las Vegas complexity.
Indeed, let $\psi_t = \Query_t(\cA, O)\xi$ be the state processed by $\tO$ on the $t$-th query, and $\psi_{t,1},\dots,\psi_{t,d}$ be the corresponding states processed by the oracle $\tO'$ on the right-hand side of~\rf{eqn:slicingA}.
Then
\begin{equation}
\label{eqn:slicingB}
\psi_t = \psi_{t,1}\oplus \psi_{t,2}\oplus \cdots \oplus \psi_{t,d},
\end{equation}
and the result follows from~\rf{eqn:DnormSum}.
\pfend

Note that the algorithm depends on the choice of slicing in~\rf{eqn:slicingA}, which in turn depends on the choice of the orthonormal basis in the space of $\Ib$.
By~\rf{eqn:slicingA}, this does not change the action of the algorithm, and by~\rf{eqn:slicingB}, this does not change its complexity.
Thus, we can further assume, without loss of generality, that a quantum algorithm is sliced.
We will use this in this section, as it simplifies some constructions and some proofs.

Also, note that the proof of \rf{prp:slicing} still works if we have different embeddings of $O$ on each query of the algorithm in \rf{defn:algorithm}.
Thus, this variant of the definition is also equivalent to \rf{defn:algorithm}.

\subsection{Space Extension}

\def\sS[#1]{\vcenter{\hbox{$\scriptstyle ($}} #1 \vcenter{\hbox{$\scriptstyle )$}}}

The following two results formally state that we can embed an algorithm into a larger space.
The work space extension is straightforward:

\begin{prp}[Work Space Extension]
\label{prp:workSpaceExtension}
Let $\cA$ be an algorithm in $\cH$ with oracles in $\cM$.
Then, for every $\cH'$, there is an algorithm $\cA\oplus I_{\cH'}$ in $\cH\oplus\cH'$ with oracles in $\cM$ such that for every $O\colon \cM\to\cM$, $\xi\in \cH$, and $\xi'\in\cH'$, we have
$\sS[\cA\oplus I_{\cH'}](O) = \cA(O)\oplus I_{\cH'}$
and
$L(\cA\oplus I_{\cH'}, O, \xi\oplus \xi') = L(\cA, O, \xi)$.
\end{prp}

\pfstart
Let $\cA$ be as in \rf{defn:algorithm}.
To get $\cA\oplus I_{\cH'}$, replace each $U_i$ with $U_i\oplus I_{\cH'}$, and each $\Ic$ from $\tO$ with $\Ic \oplus I_{\cH'}$.
\pfend


The input space extension is also possible.
For simplicity, we assume the algorithm $\cA$ is sliced.
We state the extension in a rather general way.
Essentially, we require that the input oracle in the extended space agrees with the original oracle on the states actually being queried.

Let $\cA$ be a sliced algorithm in $\cH$ with oracle $O\colon \cM\to\cM$, and $\cM\oplus \cM'$ be a superspace of $\cM$.
We construct an algorithm $\cA'$ in $\cH\oplus \cM'$ with oracle $O'\colon \cM\oplus \cM'\to\cM\oplus \cM'$ in the following way.
Each unitary $U_i$ is replaced by $U_i\oplus I_{\cM'}$ and each query $O\oplus \Ic$ is replaced by $O'\oplus \Ic$ acting in $\cH\oplus \cM'$.

\begin{prp}[Input Space Extension]
\label{prp:inputSpaceExtension}
In the above assumptions, 
if $O\colon\cM\to\cM$, $O'\colon \cM\oplus\cM'\to\cM\oplus\cM'$ and $\xi\in \cH$ are such that
\begin{equation}
\label{eqn:inputOracleConsistency}
O \Query_t(\cA, O)\xi = O'\Query_t(\cA, O)\xi
\end{equation}
for all $t$,
then
$\cA'(O')\xi = \cA(O)\xi$
and 
$L(\cA', O', \xi) = L(\cA, O, \xi)$.

In particular, Eq.~\rf{eqn:inputOracleConsistency} holds if $O' = O\oplus O''$ for some $O''$ acting in $\cM'$.
\end{prp}

\pfstart
Recall the operator $\State_t$ defined in~\rf{eqn:State}.
By induction on $t$, it is easy to show that $\State_t(\cA', O')\xi = \State_t(\cA, O) \xi$,
from which the statement follows.
\pfend

We will often identify the algorithms $\cA$ and $\cA'$ above.


\subsection{Sequential Composition and Direct Sum}

\begin{prp}[Sequential Composition]
\label{prp:sequential}
Assume there are two algorithms $\cA$ and $\cB$ in $\cH$ with oracles in $\cM$.
Then, there exists an algorithm $\cB*\cA$ such that for all $O\colon \cM\to\cM$ and $\xi\in\cH$ we have $\sS[\cB*\cA](O) = \cB(O)\cA(O)$ and 
\[
L(\cB*\cA, O, \xi) = L\sA[\cB, O, \cA(O)\xi] + L(\cA, O, \xi).
\]
\end{prp}

\pfstart
The algorithm $\cB * \cA$ is the algorithm $\cB$ applied after $\cA$.
\pfend

The condition that $\cA$ and $\cB$ share the same workspace seems restrictive, but it is necessary for the formal statement of \rf{prp:sequential}.
Usually it makes sense to assume that $\cA$ and $\cB$ share the same \emph{output} space $\cK$.
Then, in the spirit of \rf{defn:subspaceConversion}, the initial state $\xi$ is assumed to be such that both $\cA(O)\xi$ and $\sS[\cB * \cA](O)\xi$ are in $\cK$.
Let $\cW$ and $\cW'$ be the orthogonal complements of $\cK$ in the workspaces of $\cA$ and $\cB$, respectively (the ``scratch spaces'').
We can still apply \rf{prp:sequential} with $\cH = \cK\oplus \cW\oplus \cW'$ and assuming that the algorithms $\cA$ and $\cB$ are extended by the identity to $\cH$ using \rf{prp:workSpaceExtension}.

Also, \rf{prp:sequential} assumes that $\cA$ and $\cB$ use the same input oracle $O$.
This is without loss of generality.
Indeed, let $\cA$ and $\cB$ use different oracles $O'\colon\cM'\to\cM'$ and $O''\colon \cM''\to\cM''$.
Extend the input space of both algorithm to $\cM = \cM'\oplus\cM''$, and assume they both use the input oracle $O = O'\oplus O''$.
By \rf{prp:inputSpaceExtension}, the action of both algorithms does not change.
The same observations also applies to Propositions~\ref{prp:directSum} and~\ref{prp:tensorProduct} below.

\begin{prp}[Direct Sum]
\label{prp:directSum}
Let $\cA$ and $\cB$ be two algorithms in spaces $\cH$ and $\cH'$ respectively, and both with oracles in $\cM$.
Then, there exists an algorithm $\cA\oplus\cB$ in $\cH\oplus\cH'$ with oracles in $\cM$ such that for all $O\colon \cM\to\cM$, $\xi\in\cH$, and $\xi'\in\cH'$, we have $\sS[\cA\oplus \cB](O) = \cA(O)\oplus\cB(O)$ and
\begin{equation}
\label{eqn:directSum}
L(\cA\oplus\cB, O, \xi\oplus\xi') = L(\cA, O, \xi) + L(\cB, O, \xi').
\end{equation}
\end{prp}

\pfstart
The algorithm $\cA\oplus \cB$ can be implemented as $(I_\cH\oplus \cB)*(\cA\oplus I_{\cH'})$.
The result follows from Propositions~\ref{prp:workSpaceExtension} and~\ref{prp:sequential}.
\pfend

\subsection{Functional Composition and Tensor Product}
\label{sec:composition}
Functional composition is a more interesting way of composing algorithms.
It can be constructed with ease assuming the outer algorithm is sliced.

\begin{prp}[Functional Composition]
\label{prp:composition}
Let $\cA$ be a sliced algorithm in $\cH$ with oracles in $\cN$, and $\cB$ be an algorithm in $\cN$ with oracles in $\cM$.
Then, there exists an algorithm $\cA\circ \cB$ in $\cH$ with oracles in $\cM$ such that for all $O\colon \cM\to\cM$ and $\xi\in \cH$, we have
$\sS[\cA\circ\cB](O) = \cA(\cB(O))$ and
\begin{equation}
\label{eqn:composition}
L(\cA\circ\cB,O,\xi) = \sum_t L\s[{\cB, O, \Query_t \sA[\cA, \cB(O)]\xi}].
\end{equation}
\end{prp}

\pfstart
Denote by $O'\colon \cN\to\cN$ the input oracle of the outer algorithm $\cA$.
Replace each query $\tO' = O'\oplus \Ic$ of $\cA$ by a copy of the algorithm $\cB\oplus\Ic$ obtained via \rf{prp:workSpaceExtension}.
The theorem follows from \rf{prp:sequential} and the observation that the copy of the algorithm $\cB$ replacing the $t$-th query processes the state $\Query_t \sA[\cA, \cB(O)]\xi$.
\pfend

This result requires a number of comments.

First, it is usually convenient to assume that $\cN$ is the \emph{output} space of the algorithm $\cB$, not its workspace.
This can be achieved by applying \rf{prp:inputSpaceExtension}, \emph{cf.} the discussion after \rf{prp:sequential}.

Next, \rf{prp:composition} assumes the that algorithm $\cA$ has a single input oracle (while the algorithm $\cB$ and, consequently, $\cA\circ\cB$ can have multiple input oracles).
Let us now consider the case when $\cA$ has multiple input oracles $O^{(i)}\colon \cN^{(i)}\to\cN^{(i)}$.
For each $i$, let $\cB^{(i)}$ be an algorithm in $\cN^{(i)}$ with the oracle in $\cM$.
Using \rf{prp:directSum}, they can be combined into a single algorithm $\cB = \bigoplus_i \cB^{(i)}$ acting in $\cN = \bigoplus_{i} \cN^{(i)}$, which is the same space where the combined input oracle of $\cA$ acts.
Thus, using~\rf{eqn:directSum}, we obtain the following version of~\rf{eqn:composition}:
\begin{equation}
\label{eqn:compositionMultipleOracles}
L(\cA\circ\cB,O,\xi) = \sum_{i} \sum_t L\s[{\cB^{(i)}, O, \Query_t^{(i)} \sA[\cA, \cB(O)]\xi}].
\end{equation}

We will return to the above two comments in \rf{sec:compositionRevisited}.
\medskip

The final comment concerns slicing.
Namely, when applying \rf{prp:composition} to a non-sliced algorithm $\cA$ as in \rf{defn:algorithm}, it is first necessary to slice the latter using \rf{prp:slicing}.
Slicing is convenient here as it allows us to use Las Vegas complexity of $\cB$ on the state $\Query_t \sA[\cA, \cB(O)]\xi$ directly.
The downside of this approach is that the resulting algorithm depends on the way how we slice the query $O\otimes \Ib$ of the algorithm $\cA$ in~\rf{eqn:slicingA}.
As discussed before, this does not change the action of the algorithm.
However, it is not clear how it affects complexity.

In order to understand this, it suffices to consider one query of the outer algorithm.
That is, we can assume the composed algorithm is of the form $\cB\otimes \Ib$.
Applying \rf{prp:composition} to the sliced algorithm and using the following decomposition similar to~\rf{eqn:slicingB}:
\begin{equation}
\label{eqn:xiSlicing}
\xi = \xi_1 \oplus \xi_2 \oplus \cdots \oplus \xi_d
\end{equation}
with each $\xi_j$ in $\cN$, we get 
\begin{equation}
\label{eqn:algorithmTimesIdenitiy}
L\s[\cB\otimes\Ib, O, \xi] = \sum_{j=1}^d L(\cB, O, \xi_j).
\end{equation}

\begin{obs}
\label{obs:slicingIndependent}
The value of the right-hand side of~\rf{eqn:algorithmTimesIdenitiy} is independent from the choice of a particular slicing in~\rf{eqn:xiSlicing}.
\end{obs}

Therefore, for a non-sliced algorithm $\cA$, we can write an analogue of~\rf{eqn:composition}:
\begin{equation}
\label{eqn:compositionNonSliced}
L(\cA\circ\cB,O,\xi) = \sum_t L\s[{\cB\otimes\Ib, O, \Query_t \sA[\cA, \cB(O)]\xi}],
\end{equation}
which is well-defined due to the above observation.
Similarly, in the case of multiple input oracles, we can write the following analogue of~\rf{eqn:compositionMultipleOracles}:
\[
L(\cA\circ\cB,O,\xi) = \sum_{i} \sum_t L\s[{\cB^{(i)}\otimes \Ib, O, \Query_t^{(i)} \sA[\cA, \cB(O)]\xi}].
\]

\pfstart[Proof of \rf{obs:slicingIndependent}]
In~\rf{eqn:xiSlicing}, we decomposed $\xi$ assuming some standard basis in the space of $\Ib$.
Let $u_1,\dots,u_d$ be another orthonormal basis of the same space.
Thus, we have a similar decomposition
\begin{equation}
\label{eqn:xiSlicingB}
\xi = \xi'_1\otimes u_1 + \xi'_2 \otimes u_2 + \cdots + \xi'_d\otimes u_d
\end{equation}
with $\xi'_1,\dots,\xi'_d\in\cN$, but this time based on the basis $u_1,\dots,u_d$.

Since the the basis $u_1,\dots,u_d$ is orthonormal the decompositions in~\rf{eqn:xiSlicing} and~\rf{eqn:xiSlicingB} are connected by a unitary $U$ in the following way, where we assume the unitary $U$ is given by~\rf{eqn:unitaryU}:
\[
\xi'_j = \alpha_{1,j} \xi_1 + \alpha_{2,j}\xi_2 + \cdots + \alpha_{d,j} \xi_d .
\]
Therefore, the complexity of the algorithm obtained when using the slicing based on $u$ is
\[
\sum_{j=1}^d L(\cB, O, \xi'_j) 
= \sum_{j=1}^d L(\cB, O, \alpha_{1,j} \xi_1 + \alpha_{2,j}\xi_2 + \cdots + \alpha_{d,j} \xi_d)
= \sum_{j=1}^d L(\cB, O, \xi_j) 
\]
by \rf{prp:parallelogram}.
\pfend

As a by-product we get a nice expression for a tensor product of algorithms.

\begin{prp}[Tensor Product]
\label{prp:tensorProduct}
Let $\cA$ and $\cB$ be two algorithms in spaces $\cH$ and $\cH'$ respectively, and both with oracles in $\cM$.
Then, there exists an algorithm $\cA\otimes\cB$ in $\cH\otimes\cH'$ with oracles in $\cM$ such that for all $O\colon \cM\to\cM$, we have $\sS[\cA\otimes \cB](O) = \cA(O)\otimes\cB(O)$.
Moreover, if $O$ is a unitary, then
\begin{equation}
\label{eqn:tensorProduct}
L(\cA\otimes\cB, O, \xi) = L(\cA\otimes I_{\cH'}, O, \xi) + L(I_{\cH}\otimes \cB, O, \xi),
\end{equation}
where the two terms on the right-hand side are defined as in~\rf{eqn:algorithmTimesIdenitiy}.
\end{prp}

\pfstart
We can implement $\cA\otimes \cB$ as $(I_\cH\otimes \cB)*(\cA\otimes I_{\cH'})$.
By \rf{prp:sequential}, we get
\[
L\sA[\cA\otimes\cB, O, \xi] = L\sA[\cA\otimes I_{\cH'}, O, \xi] + L\sA[I_{\cH}\otimes \cB, O, {\sS[\cA(O)\otimes I_{\cH'}]}\xi].
\]
Therefore, it remains to prove that
\[
L(I_{\cH}\otimes \cB, O, \xi) = L\sA[I_{\cH}\otimes \cB, O, {\sS[\cA(O)\otimes I_{\cH'}]}\xi].
\]
But this follows from \rf{obs:slicingIndependent}, as multiplication by a unitary $\cA(O)\otimes I_{\cH'}$ can be seen as a change of basis in $\cH$.
\pfend

If $O$ is not unitary, we do not get such a nice expression as~\rf{eqn:tensorProduct}.
For instance, the complexity depends on whether we implement $\cA\otimes \cB$ as $(I_\cH\otimes \cB)*(\cA\otimes I_{\cH'})$ or as $(\cA\otimes I_{\cH'})*(I_\cH\otimes \cB)$.

\section{
\texorpdfstring{Unidirectional Relative $\gamma_2$-bound}
{Unidirectional Relative gamma2-bound}
}
\label{sec:onegamma}

The variants of the adversary bound in~\cite{lee:stateConversion} and~\cite{belovs:variations} are formulated in terms of generalisations of the $\gamma_2$-norm.
The $\gamma_2$-norm was originally developed in the context of operator factorisation in Banach spaces~\cite[Section 13]{tomczak:banach}.
It has an independent formulation as the Schur (Hadamard) product operator norm~\cite{bhatia:positive}.
In the realm of theoretical computer science, it was first used in communication complexity~\cite{linial:complexityOfSignMatrices, linial:lowerBoundsInCommunicationComplexity, lee:directDiscrepancy}.
In the context of the quantum adversary, its generalisations appeared in~\cite{lee:stateConversion} and~\cite{belovs:variations} as filtered and relative $\gamma_2$-norms, respectively.

We have to generalise the latter in several directions.
First, in order to deal with unidirectional access to the input oracle, we have to define the unidirectional version of the bound, which we do in \rf{sec:singleObjective}.
The previous (bidirectional) case can be obtained as a special case, see \rf{sec:bidirectionality}.
Second, in order to switch from the worst-case complexity to the complete complexity profile, we have to introduce the multi-objective version of the bound.
Finally, we also have to modify the bound to capture the case of several input oracles.
All this is done in \rf{sec:multiObjective}.

In \rf{sec:gammaProperties}, we prove few basic properties of the unidirectional relative $\gamma_2$-bound, which we will need later in the paper.


\subsection{Single-Objective Version}
\label{sec:singleObjective}

\begin{defn}[Unidirectional relative $\gamma_2$-bound]
\label{defn:onegamma}
Let $\cK$ and $\cM$ be vector spaces, and $D$ be a set of labels.
Let $E = \{E_{xy}\}$ and $\Delta = \{\Delta_{xy}\}$, where $x,y\in D$, be two families of linear operators: $A_{xy}\colon \cK\to\cK$ and $\Delta_{xy}\colon \cM\to\cM$ that satisfy $E_{xy} = E_{yx}^*$ and $\Delta_{xy} = \Delta_{yx}^*$ for all $x,y\in D$.

The \emph{unidirectional relative $\gamma_2$-bound}
\[
\onegamma(E | \Delta) = \onegamma(E_{xy} \mid \Delta_{xy})_{x,y\in D},
\]
is defined as the optimal value of the following optimisation problem, where $V_x$ are linear operators:
\begin{subequations}
\label{eqn:onegammaMultidimensional}
\begin{alignat}{3}
&\mbox{\rm minimise} &\quad& \max\nolimits_{x\in D} \norm|V_x|^2 &\quad&\\
& \mbox{\rm subject to}&&  
E_{xy} = V_x^* (\Delta_{xy}\otimes I_{\cW}) V_y && \text{\rm for all $x, y\in D$;}  \label{eqn:onegammaMultidimensionalCondition}
\\
&&& \text{$\cW$ is a vector space}, &&
V_x \colon \cK\to \cM\otimes\cW.
\end{alignat}
\end{subequations}
\end{defn}

Depending on the context, we will denote by $\onegamma(E|\Delta)$ both the optimal value and the optimization problem itself.

We will be mostly using the following one-dimensional version, where each $E_{x,y} = e_{x,y}$ is a scalar.  Then, the bound reads as follows:
\begin{subequations}
\label{eqn:onegamma}
\begin{alignat}{3}
&\mbox{\rm minimise} &\quad& \max\nolimits_{x\in D} \norm|v_x|^2 &\quad&\\
& \mbox{\rm subject to}&&  
e_{xy} = \ipA<v_x,\;  (\Delta_{xy}\otimes I_{\cW}) v_y> && \text{\rm for all $x, y\in D$;} \label{eqn:onegammaCondition} \\
&&& \text{$\cW$ is a vector space}, &&
v_x \in \cM\otimes\cW.
\end{alignat}
\end{subequations}

The version~\rf{eqn:onegamma} is the one mentioned in \rf{fig:correspondence} in the introduction.  Its feasible solutions correspond to the algorithms solving the problem.
In order to prove lower bounds, we need another closely related notion.
Let us define the following generalisation of the Hadamard product.
Assume $X$ and $Y$ be some sets of labels, and $\Delta = (\Delta_{x,y})$, where $x\in X$ and $y\in Y$, be a set of matrices of the same dimensions.
For $\Gamma$, an $X\times Y$ matrix, we define $\Gamma\circ \Delta$ as an $X\times Y$ block matrix, where the block corresponding to $x\in X$ and $y\in Y$ is given by $\Gamma\elem[x,y]\Delta_{x,y}$.

\begin{defn}[Unidirectional subrelative $\gamma_2$-bound]
\label{defn:subgamma}
In assumptions of \rf{defn:onegamma}, the \emph{unidirectional subrelative $\gamma_2$-bound}
\[
\subgamma(E | \Delta) = \subgamma(E_{xy} \mid \Delta_{xy})_{x,y\in D},
\]
is defined as the optimal value of the following optimisation problem:
\begin{subequations}
\label{eqn:onegammadual}
\begin{alignat}{2}
&\mbox{\rm maximise} &\quad& \lambda_{\max} (\Gamma\circ E) \\
& \mbox{\rm subject to}&&  \lambda_{\max}(\Gamma\circ \Delta)\le 1 \label{eqn:onegammadualCondition}, 
\end{alignat}
\end{subequations}
where $\Gamma$ ranges over $D\times D$ Hermitian matrices.
Here $\lambda_{\max}$ stands for the largest eigenvalue of a Hermitian matrix.
\end{defn}

This version is similar to the dual of the relative $\gamma_2$-norm from~\cite{belovs:variations}, except that it has $\lambda_{\max}$ instead of the spectral norm.  The latter, in its turn, is similar to the negative-weighted adversary from~\cite{hoyer:advNegative}.
It is easy to show that~\rf{eqn:onegammadual} lower bounds~\rf{eqn:onegammaMultidimensional}.

\begin{thm}[Weak Duality]
\label{thm:weakduality}
For $E$ and $\Delta$ as in Definitions~\ref{defn:onegamma} and~\ref{defn:subgamma}, we have $\onegamma(E|\Delta)\ge \subgamma(E|\Delta)$.
\end{thm}

\pfstart
Assume we have a feasible solution $V_x$ to $\onegamma(E|\Delta)$.
From~\rf{eqn:onegammaMultidimensionalCondition}, we get that for every $D\times D$-matrix $\Gamma$:
\[
\Gamma\circ E = V^* \skA[(\Gamma\circ\Delta)\otimes I_\cW]\,V,
\]
where $V = \bigoplus_{x\in D} V_x$ is the block-diagonal matrix with the blocks $V_x$ on the diagonal.
Hence,
\begin{align*}
\lambda_{\max}(\Gamma\circ E) 
&= \max_{v: \|v\|=1} v^* (\Gamma\circ E) v
=  \max_{v: \|v\|=1} (Vv)^* \skA[(\Gamma\circ\Delta)\otimes I_\cW]\, Vv\\
&\le \|V\|^2 \cdot \lambda_{\max} \sA[(\Gamma\circ\Delta)\otimes I_\cW] 
= \max_{x\in D}\|V_x\|^2 \cdot \lambda_{\max} \s[\Gamma\circ\Delta]
\le \onegamma(A|\Delta)\lambda_{\max}(\Gamma\circ \Delta).\qedhere
\end{align*}
\pfend

Let us note, although we will not need it in this paper, that in the one-dimensional case it is possible to strengthen the previous theorem.
\begin{thm}
\label{thm:onegammadual}
If all $E_{x,y}=e_{x,y}$ are one-dimensional, then $\onegamma(E|\Delta)=\subgamma(E|\Delta)$.
\end{thm}

Therefore, the lower bound~\rf{eqn:onegammadual} is tight in this case.
The proof follows from strong duality and is a variant of the proof in~\cite{belovs:variations}.
It can be found in \rf{app:duality}.
Note that \rf{thm:weakduality} is not true in general, when $E_{x,y}$ are not one-dimensional.

\subsection{Multi-Objective Version}
\label{sec:multiObjective}

Since we consider Las Vegas complexity of each individual input, considering a single number as an output of an optimisation problem like~\rf{eqn:onegammaMultidimensional} and~\rf{eqn:onegamma} is too restrictive.
Here we define the multi-objective version of the same optimisation problem.
Additionally, we consider the version of the bound with multiple $\Delta$, which corresponds to the multiple-oracle case of \rf{sec:multipleOracles}.

For the latter, assume that $\cM$ from \rf{defn:onegamma} is decomposed as $\cM = \cM^{(1)}\oplus \cM^{(2)}\oplus \cdots \oplus \cM^{(s)}$, and each $\Delta_{xy}$ has a similar decomposition:
\begin{equation}
\label{eqn:DeltaDecomposition}
\Delta_{xy} = \Delta^{(1)}_{xy} \oplus \Delta^{(2)}_{xy}\oplus \cdots\oplus \Delta^{(s)}_{xy}
\end{equation}
with $\Delta^{(i)}_{xy}\colon \cM^{(i)}\to\cM^{(i)}$.
Additionally, we write $V_x\colon \cK \to \cM\otimes \cW$ from~\rf{eqn:onegammaMultidimensional} as a vertical stack of matrices
\[
V_x = \begin{pmatrix}
V_x^{(1)}\\V_x^{(2)}\\\vdots\\V_x^{(s)}
\end{pmatrix}
\]
with $V_x^{(i)}\colon \cK\to \cM^{(i)}\otimes \cW$.
We also generalise~\rf{eqn:Dnorm} to such matrices:
\[
\Dnorm|V_x|^2 = \sB[ \|V_x^{(1)}\|^2,  \|V_x^{(2)}\|^2, \dots, \|V_x^{(s)}\|^2 ].
\]

\begin{defn}[Multi-objective unidirectional relative $\gamma_2$ optimisation problem]
In notation of \rf{defn:onegamma} and the above assumptions on $\cM$ and $\Delta$, the \emph{multi-objective unidirectional relative $\gamma_2$ optimisation problem} 
\[
\onegamma(E | \Delta) = \onegamma(E_{xy} \mid \Delta_{xy})_{x,y\in D},
\]
is defined as follows:
\begin{subequations}
\label{eqn:onegammaMultiobjective}
\begin{alignat}{3}
&\mbox{\rm minimise} &\quad& (\Dnorm|V_x|^2)_{x\in D} &\quad&\\
& \mbox{\rm subject to}&&  
E_{xy} = V_x^* (\Delta_{xy}\otimes I_{\cW}) V_y && \text{\rm for all $x, y\in D$;}  \label{eqn:onegammaMultiobjectiveCondition}
\\
&&& \text{$\cW$ is a vector space}, &&
V_x \colon \cK\to \cM\otimes\cW.
\end{alignat}
\end{subequations}
\end{defn}

The bound~\rf{eqn:onegammaMultiobjective} is equivalent to~\rf{eqn:onegammaMultidimensional} with the only difference in the objective, which justifies the use of the same notation $\onegamma(E|\Delta)$.
Later we will almost exclusively use the multi-objective version. 

In the multiple-oracle case, we assume that the decomposition in~\rf{eqn:DeltaDecomposition} is implicit.  
Note that it only changes the objective, and does not change the constraints.
Also, the constraint~\rf{eqn:onegammaMultiobjectiveCondition} in this case is equivalent to
\[
E_{xy} = \sum_{i=1}^s \sA[V_x^{(i)}]^* \sA[\Delta_{xy}^{(i)}\otimes I_{\cW}] V_y^{(i)}.
\]

\begin{defn}
\label{defn:feasibleObjective}
For a feasible solution $V_x$ of~\rf{eqn:onegammaMultiobjective}, we call $(\Dnorm|V_x|^2)_{x\in D}$ the \emph{objective profile} of the feasible solution.
In the single-oracle case, it is a vector in $\bR^D$.
In the multiple-oracle case, it is a vector in $\bR^D\otimes \bR^s$.
The \emph{feasible objective space} of the optimization problem~\rf{eqn:onegammaMultiobjective} is the set of all objective profiles over all feasible solutions $V_x$ of~\rf{eqn:onegammaMultiobjective}.
\end{defn}

\begin{clm}
\label{clm:feasibleClosed}
If all $E_{x,y}=e_{x,y}$ are one-dimensional, the feasible objective space of~\rf{eqn:onegammaMultiobjective} is a topologically closed subset of $\bR^D\otimes \bR^s$.
\end{clm}

This is also true in general, but we only need the one-dimensional case, which we prove in \rf{app:duality}.

Finally, for the multi-oracle case, we have the following variant of \rf{thm:weakduality}, which binds the matrix $\Gamma$ to the individual $\|V_x^{(i)}\|^2$.
It can be used to prove trade-offs between input oracles.

\begin{thm}
For every feasible solution $V_x$ to~\rf{eqn:onegammaMultiobjective} and every $D\times D$ Hermitian matrix $\Gamma$, we have
\[
\lambda_{\max} (\Gamma\circ E) \le \sum_{i=1}^s \lambda_{\max}(\Gamma\circ\Delta^{(i)}) \max_{x\in D} \normA|V_x^{(i)}|^2.
\]
\end{thm}

\pfstart
Again, let $V = \bigoplus_{x\in D} V_x$ and $V^{(i)} = \bigoplus_{x\in D} V_x^{(i)}$.
The proof follows the proof of \rf{thm:weakduality} with the following change at the last step:
\begin{align*}
\lambda_{\max}(\Gamma\circ E) 
&= \max_{v: \|v\|=1} v^* (\Gamma\circ E) v
=  \max_{v: \|v\|=1} (Vv)^* \skA[(\Gamma\circ\Delta)\otimes I_\cW]\, Vv\\
&= \max_{v: \|v\|=1} \sum_{i=1}^s (V^{(i)}v)^* \skA[(\Gamma\circ\Delta^{(i)})\otimes I_\cW]\, V^{(i)}v\\
&\le \sum_{i=1}^s \lambda_{\max}(\Gamma\circ\Delta^{(i)}) \max_{x\in D} \normA|V_x^{(i)}|^2.\qedhere
\end{align*}
\pfend

\subsection{Properties}
\label{sec:gammaProperties}

Let us list some properties of the unidirectional relative $\gamma_2$-bound.
We are mostly interested in the case when the right-hand side $\Delta_{x,y}$ is fixed, and the left-hand side $e_{x,y}$ is variable.

\begin{prp}
\label{prp:onegammaLinearity}
Assume that $w$ and $w'$ are in the feasible objective spaces of optimization problems $\onegamma\sA[e_{x,y} | \Delta_{x,y}]_{x,y\in D}$ 
and 
$\onegamma\sA[e'_{x,y} | \Delta_{x,y}]_{x,y\in D}$, respectively.
Then, for all real $c, c'\ge 0$, the vector $cw + c' w'$ is in the feasible objective space of $\onegamma\sA[c_1e^{(1)}_{x,y} + c_2e^{(2)}_{x,y} \mid \Delta_{x,y}]_{x,y\in D}$.
\end{prp}

\pfstart
Assume that 
$\sA[v_x]_{x\in D}$ is a feasible solution to $\onegamma\sA[e_{x,y} \mid \Delta_{x,y}]_{x,y\in D}$
with objective profile $w$, and $v'_x$ is defined similarly for $w'$.
Then, 
$\sA[\sqrt{c}v_x \oplus \sqrt{c'}v'_x]_{x\in D}$ is a feasible solution to $\onegamma\sA[c e_{x,y} + c'e'_{x,y} \mid \Delta_{x,y}]_{x,y\in D}$
with objective profile $cw + c' w'$ by~\rf{eqn:DnormScaling} and~\rf{eqn:DnormSum}.
\pfend

We will often have that $\Delta_{xx}=0$ for all $x\in D$.
In this case, it is easy to specify all families of $e_{x,y}$ that have a feasible solution.

\begin{prp}
\label{prp:onegammaFeasible}
Assume that $\Delta_{x,y}$ in addition to $\Delta_{x,y}=\Delta_{y,x}^*$ satisfy $\Delta_{x,x}=0$ for all $x$.
Let $(e_{x,y})_{x,y\in D}$ be any collection of complex numbers such that $e_{x,y}=e^*_{y,x}$ for all $x,y\in D$, and $e_{x,y}=0$ whenever $\Delta_{x,y}=0$.
Then the optimisation problem $\onegamma\sA[e_{x,y}\mid \Delta_{x,y}]_{x,y\in D}$ has a feasible solution.
\end{prp}

\pfstart
Due to \rf{prp:onegammaLinearity}, it suffices to consider the case when there exist distinct $x_0, y_0\in D$ such that $e_{x_0,y_0} = e_{y_0,x_0}^*$ are the only non-zero $e_{x,y}$.
By the assumption, $\Delta_{x_0,y_0}\ne 0$.
Hence, there exist vectors $u,v$ such that $u^*\Delta_{x_0,y_0}v = 1$.
Define the feasible solution as $v_{x_0} = u$, $v_{y_0} = e_{x_0,y_0} v$, and $v_{x}=0$ otherwise.
\pfend

Describing the set of $e_{x,y}$ that have feasible solution in the general case (when $\Delta_{x,x}\ne 0$) is more complicated, and we do not do it here.

\begin{prp}
\label{prp:upwardsClosed}
If $\Delta_{x,x}=0$ for all $x$, then the feasible objective space of $\onegamma(E|\Delta)$ is upwards closed, i.e., if $w\in \bR^D\otimes \bR^s$ is in the feasible objective space, and $w'\ge w$ (component-wise), then $w'$ is also in the feasible objective space.
\end{prp}

\pfstart
Let $V_x$ be a feasible solution such that $\Dnorm|V_x|=w_x$ for all $x$.
Let $\cM_x$ be pairwise orthogonal copies of $\cM$ that are also orthogonal to $\cM\otimes\cW$.
There exist $V'_x\colon \cK\to(\cM\otimes\cW)\oplus \cM_x$ such that their projection to $\cM\otimes\cW$ agree to $V_x$ and $\Dnorm|V'_x| = w'_x$.

They also satisfy the constraints~\rf{eqn:onegammaMultiobjectiveCondition} with the properly enlarged $\cW$.
Indeed, for $x\ne y$ this follows from the orthogonality of $\cM_x$ and $\cM_y$, and for $x=y$ this follows from $\Delta_{x,x}=0$.
\pfend

\section{Adversary Bound for State Conversion}
\label{sec:advStateConversion}
This is the central section of the paper, in which we define the adversary bound for state conversion with general input oracles, and prove that it equals Las Vegas complexity.
In \rf{sec:StateConversionDef}, we restate the state conversion problem, define its Las Vegas complexity, and formulate the corresponding adversary optimisation problem.
In \rf{sec:intuition}, we explain the intuition behind the latter definition.
Sections~\ref{sec:mainlower} and~\ref{sec:mainupper} are devoted to the two main technical results: a lower bound for exact, and an upper bound for approximate state conversion.
They are the cornerstones of what comes next.
In \rf{sec:equality}, we prove an upper bound for exact state conversion, thus showing that the adversary bound is precisely equal to Las Vegas complexity.
We finish the section with two examples.
In \rf{sec:ExampleOfTwoLabels}, we consider a simple example of a state conversion problem with $|D|=2$, and in \rf{sec:booleanFunction} we obtain the adversary bound of Boolean function evaluation.

The main results in Sections~\ref{sec:mainlower} and~\ref{sec:mainupper} hold even for general linear input oracles.  In \rf{sec:equality}, we have to assume that the input oracles are unitary.

\subsection{Definitions}
\label{sec:StateConversionDef}

Our choice of problem for this section is state conversion with general input oracles.
The motivation for this initial choice is as follows.
First, we want more control on the input oracle: we require that, for each $x\in D$, we have only one input oracle.
Second, we would like to have larger flexibility on the side of the algorithm, that is why we choose the state conversion problem, where we have to map one state into another.
In the beginning, we even do it approximately.
Going to more specific tasks, like state generation, does not give us anything.
We will extend the output and the input conditions to subspace conversion in the next section.

Let us give an explicit definition of state conversion, which follows from the general consideration of \rf{sec:conditions}.

\begin{defn}[State Conversion with General Input Oracles]
\label{defn:stateConversion}
Let $D$ be a set of labels, and $\cM$ and $\cK$ vector spaces.
For each $x\in D$, let $O_x\colon \cM\to\cM$ be a linear transformation.
A \emph{state conversion problem} is given by a collection of tuples $\xi_x\mapsto \tau_x$ where $x$ ranges over $D$ and $\xi_x,\tau_x\in \cK$.
Assume that $\cK$ is embedded in the space $\cH$ of a quantum algorithm $\cA$.
We say that the algorithm $\cA$ \emph{solves} the state conversion problem $\xi_x\mapsto \tau_x$ on input oracles $O_x$, if $\cA(O_x)\xi_x = \tau_x$ for all $x\in D$.
\end{defn}

Some of the results in this section hold even if we only assume that $O_x$ are linear transformations.
However, we will usually assume that $O_x$ are contractions or unitaries.
The lower bound result hold even for infinite $D$, but for the upper bounds it is crucial that $D$ is finite.

This definition also includes the case of multiple input oracles as described in \rf{sec:multipleOracles}.
Then, as in~\rf{eqn:combinedOracle}, each $O_x = O^{(1)}_x\oplus O^{(2)}_x\oplus \cdots O^{(s)}_x$ with $O_x^{(i)}$ acting in $\cM^{(i)}$.

We derive Las Vegas complexity of this problem from the general definition of \rf{sec:LasVegas}.
We study not only the worst-case complexity, but consider each input $x\in D$ individually.

\begin{defn}[Las Vegas complexity of State Conversion]
\label{defn:LAStateConversion}
Assume we have a state conversion problem is as in \rf{defn:stateConversion}, and an algorithm $\cA$ that solves it.
The Las Vegas complexity of the algorithm $\cA$ on input $x\in D$, is defined as
$L_x(\cA) = L(\cA, O_x, \xi_x)$.
The \emph{worst-case Las Vegas complexity} is defined as $\max_{x\in D} L_x(\cA)$, in which case, we assume we have a single input oracle.
The \emph{complexity profile} of the algorithm $\cA$ is the vector in $\bR^D\otimes \bR^s$ given by $L_D(\cA) = (L_x(\cA))_{x\in D}$.
The \emph{feasible complexity space} of the problem is a subset of $\bR^D\otimes \bR^s$ which is the set of all complexity profiles of the algorithms solving the problem.
\end{defn}


Let us now define the adversary optimisation problem corresponding to the state conversion problem.
It is a generalisation of the version of the adversary bound from~\cite{belovs:variations} to the case of unidirectional input oracles.

\begin{defn}[Adversary Optimisation Problem]
Assume $\xi_x\mapsto \tau_x$ is a state conversion problem with unidirectional input oracles $O_x\colon \cM\to\cM$, as $x\in D$.
Its \emph{adversary optimisation problem} is the following unidirectional $\gamma_2$ optimisation problem:
\begin{equation}
\label{eqn:adversary}
\onegamma\sB[\ip<\xi_x, \xi_y> - \ip<\tau_x,\tau_y> \mid I_\cM - O_x^*O_y]_{x,y\in D}.
\end{equation}
\end{defn}

Here $I_\cM$ stands for the identity on $\cM$, but we often omit this subscript.
It is easy to see that the constraints of \rf{defn:onegamma} are satisfied, and this is a legitimate unidirectional $\gamma_2$-optimisation problem.
Let us write it down explicitly as we will be using it quite extensively.
We consider it as a multi-objective optimisation problem.
\begin{subequations}
\label{eqn:advExplicit}
\begin{alignat}{3}
&\mbox{\rm minimise} &\quad& \sA[\Dnorm|v_x|^2]_{x\in D} &\quad&\\
& \mbox{\rm subject to}&&  
\ip<\xi_x, \xi_y> - \ip<\tau_x,\tau_y> = \ipA<v_x,\;  ((I-O^*_xO_y)\otimes I_{\cW}) v_y> && \text{\rm for all $x, y\in D$;}  \label{eqn:advExplicitCondition}\\
&&& \text{$\cW$ is a vector space}, \qquad
v_x \in \cM\otimes\cW.
\end{alignat}
\end{subequations}

\subsection{Intuition}
\label{sec:intuition}
Let us describe the intuition behind the bound~\rf{eqn:adversary}.
For a collection of vectors $(\xi_x)_{x\in D}$, let $G_\xi$ denote the corresponding Gram matrix: $G_\xi\elem[x,y] = \ip<\xi_x, \xi_y>$.
Two collections of vectors can be transformed one into another by a unitary transformation if and only if they have the same Gram matrix.
Since unitary transformations are free in quantum query algorithms, we may replace collections of vectors by the corresponding Gram matrices.
For instance, rather than saying that an algorithm solves state conversion $\xi_x\mapsto\tau_x$, we can say that it transforms $G_\xi$ into $G_\tau$, or write $G_\tau\mapsto G_\xi$.

Then, the left-hand side of~\rf{eqn:advExplicitCondition} gives the difference of the corresponding Gram matrices $G_\xi - G_\tau$.
The right-hand side
\begin{equation}
\label{eqn:SuperQuery}
\ipA<v_x,\;  ((I - O^*_xO_y)\otimes I_{\cW}) v_y> =
\ipA<v_x, v_y> -
\ipA<(O_x\otimes I_{\cW}) v_x,\;  (O_y\otimes I_{\cW}) v_y>
\end{equation}
gives the change in the Gram matrix when the state $v_x$ is processed by the oracle $O_x$.
The objective value $\Dnorm|v_x|^2$ can be interpreted as the corresponding Las Vegas complexity.
Therefore, the optimisation problem seeks for the best possible states $v_x$ to be processed by the oracle to get the required change in the Gram matrix.
The issue of how to get the states $v_x$ to the input oracle is ignored here.
Therefore, one can see the adversary optimisation problem as a semi-definite relaxation of a quantum query algorithm.

To get the lower bound in \rf{sec:mainlower}, we accumulate changes in the Gram matrix like in~\rf{eqn:SuperQuery} over all the queries made by the algorithm.
The proof closely follows the proof of Theorem 10 from~\cite{belovs:variations}.
To get the algorithm in \rf{sec:mainupper}, we repeatedly apply a scaled down version of the query in~\rf{eqn:SuperQuery}.
The result is that the Gram matrix slowly slides close to the line connecting $G_\xi$ to $G_\tau$ in the cone of $D\times D$ semi-definite matrices.

\subsection{Lower Bound (For Exact Version)}
\label{sec:mainlower}

\begin{thm}
\label{thm:mainlower}
Assume $\cA$ is an algorithm that performs state conversion $\xi_x\mapsto \tau_x$ with unidirectional access to general linear oracles $O_x$ as $x\in D$.  
Then, its complexity profile $L_D(\cA)$ is in the feasible objective space of the adversary optimization problem~\rf{eqn:adversary}.
\end{thm}

\pfstart
Denote for brevity
\[
\psi_{t,x} = \State_t(\cA, O_x)\xi_x =  U_{t-1}\, \tO_x\, U_{t-2}\,\tO_x\,\cdots U_{1}\, \tO_x U_0\xi_x,
\]
and let 
$
\psi'_{t,x} = \Query_t(\cA, O_x)\xi_x
$
be the state processed on step $t$ by the input oracle.
The operators $\State_t$ and $\Query_t$ are defined in~\rf{eqn:State} and~\rf{eqn:Query}, respectively.

We have $\ip<\psi_{1,x}, \psi_{1,y}> = \ip<\xi_x, \xi_y>$, and
$\psi_{T+1, x} = \tau_x$.  This gives
\[
\ip<\xi_x, \xi_y> - \ip<\tau_x, \tau_y>
=
\sum_{t=1}^T \sB[\ip<\psi_{t,x}, \psi_{t,y}> - \ip<\psi_{t+1,x}, \psi_{t+1,y}>].
\]
Next, for the effect of one query $\tO_x = (O_x\otimes\Ib)\oplus \Ic$:
\begin{equation}
\label{eqn:mainderivation}
\begin{aligned}
\ip<\psi_{t,x}, \psi_{t,y}> - \ip<\psi_{t+1,x}, \psi_{t+1,y}>
&= \ip<\psi_{t,x}, \psi_{t,y}> - \ip<\tO_x\psi_{t,x}, \tO_y\psi_{t,y}> \\
&= \ip<\psi_{t,x}', \psi_{t,y}'> - \ipA<(O_x\otimes \Ib)\psi'_{t,x}, (O_y\otimes \Ib)\psi'_{t,y}>\\
&= \ip<\psi_{t,x}', \psi_{t,y}'> - \ipA<\psi'_{t,x}, (O_x^*O_y\otimes \Ib)\psi'_{t,y}>\\
&= \ipA<\psi'_{t,x}, ((I_{\cM}-O_x^*O_y)\otimes \Ib)\psi'_{t,y}>.
\end{aligned}
\end{equation}
This means that we can take 
\begin{equation}
\label{eqn:v_x}
v_x = \bigoxplus_{t=1}^T \psi'_{t,x}
\end{equation}
as a feasible solution to~\rf{eqn:adversary}.
By~\rf{eqn:LasVegasMultipleOracles} and~\rf{eqn:DnormSum}, $\Dnorm|v_x|^2$ is equal to the Las Vegas complexity $L_x$, hence, this feasible solution has $L_D(\cA)$ as its objective profile.
\pfend

The theorem is proven for exact and coherent state conversion.
However, it can be used for approximate or non-coherent state conversion $\xi_x\mapsto\tau_x$ as well.
Indeed, the latter is equivalent to exact coherent state conversion $\xi_x\mapsto \tau_x'$ for \emph{some} $\tau'_x$ satisfying the corresponding closeness requirements to $\tau_x$ as explained in \rf{sec:conditions}.
See \rf{sec:permutationInversion} for an example.

\mycutecommand\TotalQuery{\cV}
The map $\xi_x\mapsto v_x$ is important enough, so that we introduce a special notation for it:
\begin{equation}
\label{eqn:TotalQuery}
\TotalQuery(\cA, O)\colon \xi\mapsto \bigoplus_{t=1}^T \Query_t(\cA, O)\xi,
\end{equation}
which is a linear transformation.

\subsection{Upper Bound (For Approximate Version)}
\label{sec:mainupper}

\begin{thm}
\label{thm:mainupper}
Let $\xi_x\mapsto \tau_x$ be a state conversion problem and $O_x$ a general linear oracle, where $x$ ranges over a finite set $D$.
Assume $(v_x)_{x\in D}$ is a feasible solution to the adversary optimization problem~\rf{eqn:adversary}
with $L = \max_{x\in D} \|v_x\|^2$.
Then, for every $\eps>0$, there exists an algorithm $\cA$ with the following properties:
\itemstart
\item it solves state conversion $\xi^+_x \mapsto \tau^+_x$ with unidirectional access to $O_x$, where $\xi^+_x$ and $\tau^+_x$ are some states (not necessarily in $\cK$) satisfying $\|\xi_x^+ - \xi_x\|, \|\tau^+_x-\tau_x\|\le \eps$ for all $x\in D$;
\item its Monte Carlo query complexity is $T = \ceil[L/\eps^2]$;
\item for each $x$, its Las Vegas query complexity $L_x$ is $\Dnorm|v_x|^2$.
\itemend
Specifically, the algorithm transforms 
\begin{equation}
\label{eqn:mainupperConversion}
\xi^+_x = \xi_x \oplus \frac1{\sqrt T} v_x
\quad\longmapsto\quad 
\tau^+_x = \tau_x \oplus \frac1{\sqrt T} v_x
\end{equation}
in $T$ queries, and the state processed by the input oracle on each query is $v_x/\sqrt{T}$.
\end{thm}

\pfstart
Denote $v'_x = (O_x\otimes I_\cW) v_x$.  As in~\rf{eqn:SuperQuery}, we obtain
\[
\ip<\xi_x, \xi_y> - \ip<\tau_x, \tau_y>
= \ipA<v_x, ((I-O_x^*O_y)\otimes I_\cW)v_y>
= \ip<v_x, v_y> -  \ip<v_x',v_y'>
\]
for all $x,y\in D$.
This is equivalent to
\[
\ip<v_x',v_y'> + \ip<\xi_x, \xi_y> 
= \ip<v_x, v_y> + \ip<\tau_x, \tau_y>.
\]
This means that there exists a unitary transformation $U$ satisfying
\[
U(v'_x \oxplus \xi_x) = v_x \oxplus \tau_x
\]
for all $x\in D$.

Let us now describe the algorithm.
It depends on an integer parameter $T$, which is also its query complexity.
Its space is of the form $\cV\oplus \cK\otimes\cJ$, where $\cV$ is isomorphic to $\cM\otimes \cW$ and $\cJ$ is an $T$-qudit.

Up to unitaries, the transformation in~\rf{eqn:mainupperConversion} is equivalent to
\begin{equation}
\label{eqn:mainUpperTransform}
\frac1{\sqrt T} \ket \cV|v_x> + \ket \cK |\xi_x> \otimes \sC[\frac{1}{\sqrt T}\sum_{j=1}^{T} \ket \cJ|j>]
\quad\longmapsto\quad
\frac1{\sqrt T} \ket \cV|v_x> + \ket \cK |\tau_x> \otimes \sC[\frac{1}{\sqrt T}\sum_{j=1}^{T} \ket \cJ|j>].
\end{equation}
The algorithm performs this transformation by going through the states
\begin{equation}
\label{eqn:mainUpperA}
\psi_{t,x} = \frac1{\sqrt T} \ket \cV|v_x> 
+ \ket \cK |\tau_x> \otimes \sC[\frac{1}{\sqrt T}\sum_{j=1}^{t-1} \ket \cJ|j>]
+ \ket \cK |\xi_x> \otimes \sC[\frac{1}{\sqrt T}\sum_{j=t}^{T} \ket \cJ|j>]
\end{equation}
just before the $t$-th query.
Note that $\psi_{1,x}$ and $\psi_{T+1,x}$ are the states on the left- and the right-hand sides of~\rf{eqn:mainUpperTransform}, respectively.
On the $t$-th query, apply the input oracle $O_x\otimes I_\cW$ to the register $\cV$ in $\psi_{t,x}$.
This results in the state
\[
\frac1{\sqrt T} \ket \cV|v_x'> 
+ \ket \cK |\tau_x> \otimes \sC[\frac{1}{\sqrt T}\sum_{j=1}^{t-1} \ket \cJ|j>]
+ \ket \cK |\xi_x> \otimes \sC[\frac{1}{\sqrt T}\sum_{j=t}^{T} \ket \cJ|j>].
\]
Next, apply $U$ to the space $\cV \oplus \cK\otimes\ket \cJ|t>$, which gives $\psi_{t+1,x}$.
After $T$ iterations, we get the required transformation.

The differences $\xi^+_x-\xi_x$ and $\tau^+_x - \tau_x$ are both $v_x/\sqrt{T}$.
The norm of this vector is less than $\eps$ as long as $T \ge L/\eps^2$, as required.
Finally, the Las Vegas complexity on input $x$ is exactly
\[
T \cdot \Dnorm|\frac{v_x}{\sqrt T}|^2 = \Dnorm|v_x|^2.\qedhere
\]
\pfend

\myfigure{\label{fig:algorithm}}
{Visualisation of the algorithm $\cA$ used in \rf{thm:mainupper} and \rf{cor:main}.
The algorithm follows the straight line from $G_{\xi^+}$ to $G_{\tau^+}$, as the Gram matrices of the states in~\rf{eqn:mainUpperA} for various $t$ are uniformly placed on this line.
The wiggly line indicates the application of $\cA$ to $G_\xi$ in \rf{cor:main}.
The algorithm follows closely to the line connecting $G_{\xi^+}$ and $G_{\tau^+}$ and terminates in a point $G_{\tau'}$ close to $G_{\tau^+}$, but not, generally, $G_\tau$.
}
{
\negbigskip
\[
\begin{tikzpicture}
[
>={Latex[length=2.5mm]},
inner sep=1.5pt,
tochka/.style={circle, fill=black, radius=1pt}
]
\node [tochka, label=above:$G_{\xi^+}$] (xiplus) at (0, 0)  {} ;
\node [tochka, label=below right:$G_{\xi}$] (xi) at (0.25, -0.25)  {} ;
\node [tochka, label=above right:$G_{\tau^+}$] (tauplus) at (7, 2)  {} ;
\node [tochka, label=below right:$G_{\tau}$] (tau) at (7.25, 1.75)  {} ;
\node [tochka, label=above left:$G_{\tau'}$] (realend) at (6.7, 2.2)  {} ;
\draw [->] (xiplus) to (tauplus);
\draw [->, decorate, decoration={zigzag, amplitude=.5mm}] (xi) .. controls (2,1.2) and (5, 0.9)  .. (realend);
\end{tikzpicture}
\]
}

An immediate corollary is that for contraction oracles we can replace $\xi^+_x$ with the original $\xi_x$ and get essentially the same bound on Monte Carlo complexity.
\begin{cor}
\label{cor:main}
Assume the premises of \rf{thm:mainupper}, where the input oracles $O_x$ are contractions.
Then, for every $\eps>0$, there exists a quantum algorithm with Monte Carlo query complexity $\ceil[4L/\eps^2]$ that $\eps$-approximately and coherently solves state conversion $\xi_x \mapsto \tau_x$ with unidirectional access to $O_x$.
\end{cor}

This is a unidirectional version of the main technical result of~\cite{belovs:variations}.
This version has slightly better dependence on $\eps$, compared to~\cite{belovs:variations}, which had $O\s[\eps^{-2}\log\frac1\eps]$.
By the example due to Kothari~\cite{kothari:SDPCharacterization}, see~\cite{belovs:variations}, the dependence on $\eps$ is tight up to constant factors.

\pfstart[Proof of \rf{cor:main}]
Consider the same algorithm $\cA$ as in \rf{thm:mainupper}.
Denote by $\tau'_x$ its final state when executed on the initial state $\xi_x$ and the oracle $O_x$. See \rf{fig:algorithm} for an illustration.

Since $O_x$ is a contraction, the whole algorithm $\cA(O_x)$ is a contraction as well.
Hence,
\[
\|\tau'_x - \tau^+_x\| = \|\cA(O_x)\xi_x - \cA(O_x)\xi^+_x\| \le \|\xi_x - \xi^+_x\| \le \eps.
\]
Since $\|\tau^+_x -\tau_x\|\le \eps$, the triangle inequality gives us $\|\tau_x - \tau'_x\|\le 2\eps$.
Dividing $\eps$ by 2, we get the required algorithm.
\pfend

We also get a relation between Las Vegas and Monte Carlo complexities, which is a direct consequence of \rf{thm:mainlower} and \rf{cor:main}.

\begin{cor}
Assume there is an algorithm that solves state conversion $\xi_x\mapsto\tau_x$ with contraction input oracles $O_x$ exactly and has worst-case Las Vegas complexity $L$.
Then, there exists an algorithm that $\eps$-approximately and coherently solves the same problem and has Monte Carlo complexity $O(L/\eps^2)$.
\end{cor}

Since the distance $\|\tau_x - \tau'_x\|\le\eps$ converts to error $\eps^2$ after measurement, it is reasonable to say that the complexity of the algorithm is inversely linear in the error, which is similar to the randomised case.
This is the result mentioned in the introduction.

\subsection{Upper Bound For Exact Version}
\label{sec:equality}

The results of the previous two subsections are good enough for most purposes.
In particular, in \rf{thm:mainupper}, we can take $\xi^+_x$ and $\tau^+_x$ as close to $\xi_x$ and $\tau_x$ as we want and the Las Vegas complexity stays $\Dnorm|v_x|^2$.
However, in this section we improve this result and show how to perform \emph{exact} state conversion $\xi_x\mapsto \tau_x$ essentially in the same budget.
Together with the lower bound, \rf{thm:mainlower}, this shows that Las Vegas complexity is exactly equal to the adversary bound.

This is not only mathematically more satisfying and follows the convention of Las Vegas complexity to describe exact computation.
One of the motivations behind this result is that \emph{a priori} there is no good way to bind Las Vegas complexity on close initial states.  
If $\xi_x$ and $\xi^+_x$ are at a distance $\eps$, then, from general principles, it can be deduced that their Las Vegas complexities differ by at most $\eps T$.
But this is useless because $T$ generally is not bounded.
This poses problems if, for example, we want to compose Las Vegas programs using Propositions~\ref{prp:sequential} or~\ref{prp:composition} and we only have approximate versions of the subroutines.
\medskip


In this section, we extensively use the language of Gram matrices introduced in \rf{sec:intuition}.
Our first observation is that if both Gram matrices $G_\xi$ and $G_\tau$ are full rank, then state conversion can be performed exactly.

\begin{lem}
\label{lem:exactForPositive}
Assume the premises of \rf{thm:mainupper}.
If we additionally have that $G_\xi$ and $G_\tau$ are positive definite, then state conversion $\xi_x\mapsto \tau_x$ can be solved exactly with Las Vegas complexity $\Dnorm|v_x|^2$.
Moreover, the state processed by the oracle on each query is $v_x/\sqrt{T}$, where $T$ is the number of queries.
\end{lem}

\pfstart
Since $G_\xi, G_\tau\succ 0$, there exists an integer $T$ such that $G_\xi, G_\tau\succeq \frac{1}{T} G_v$.
Choose $\xi^-_x$ and $\tau^-_x$ so that their Gram matrices are $G_{\xi^-}=G_\xi-\frac{1}{T} G_v$ and $G_{\tau^-}=G_\tau-\frac{1}{T} G_v$. 
Note that, for all $x,y\in D$:
\[
\ipA<\xi_x^-, \xi_y^-> - \ipA<\tau_x^-, \tau_y^-> =
\ip<\xi_x, \xi_y> - \ip<\tau_x, \tau_y>.
\]
Hence, $v_x$ is a feasible solution to the adversary bound~\rf{eqn:adversary} for state conversion $\xi^-_x\mapsto \tau^-_x$ as well.
Applying \rf{thm:mainupper} to the latter and using~\rf{eqn:mainupperConversion}, we get an algorithm performing exact state conversion 
\[
\xi'_x = \xi_x^- \oplus \frac1{\sqrt T} v_x
\quad\longmapsto\quad 
\tau'_x = \tau_x^- \oplus \frac1{\sqrt T} v_x,
\]
where on each step the state $v_x/\sqrt{T}$ is processed by the oracle.
These collections of vectors have Gram matrices $G_\xi$,  and $G_\tau$, respectively.
Hence, they can be turned into $\xi_x$ and $\tau_x$, respectively, by unitaries.

In the illustration in \rf{fig:algorithm}, the $\xi_x$ and $\tau_x$ are equivalent to $\xi^+_x$ and $\tau^+_x$, and the algorithm follows the straight line connecting $G_{\xi^+}$ and $G_{\tau^+}$.
\pfend

The assumptions $G_\xi\succ 0$ and $G_\tau\succ 0$ are very strong, and almost never hold.
For instance, the state generation problem has $G_\xi$ of rank 1.
For (exact and coherent) Boolean function evaluation, the rank of $G_\tau$ is 2.
However, we will be able to apply this lemma by first ``pushing'' both $G_\xi$ and $G_\tau$ into the space of positive-definite matrices.

But for that we will need some additional assumptions.
First, we assume that all $O_x$ are unitaries.
Second, we get the point $(\|v_x\|^2)_{x\in D}$ only as a limit of points in the feasible complexity space.
This is because pushing $G_\xi$ and $G_\tau$ takes complexity, which can be made arbitrary small, but cannot be made zero.
In \rf{sec:ExampleOfTwoLabels}, we will show that the above two assumptions are necessary.
Finally, for simplicity we assume that all $O_x$ are pairwise distinct.
We will lift this restriction in \rf{sec:linearConsistency}.

Let us start with the question when a state conversion $\xi_x\mapsto\tau_x$ is possible at all.
If there is an algorithm performing the transformation, we say that $G_\tau$ is \emph{achievable} from $G_\xi$.
Denote by $\cR_\xi$ the real affine space of $D\times D$ Hermitian matrices $A$ satisfying $A\elem[x,x] = \|\xi_x\|^2$ for all $x\in D$.

\begin{lem}
\label{lem:cH_xi}
Assuming that the input oracles $O_x$ are unitary are pairwise distinct, we have:
\begin{itemize}
\itemsep0pt
\item[(a)] If the state conversion $\xi_x\mapsto \tau_x$ is possible, then $G_\tau \in \cR_\xi$ and $\cR_\tau = \cR_\xi$;
\item[(b)] For any two $M, M'\in\cR_\xi$, the optimisation problem
\[
\onegamma\sA[{M\elem[x,y] - M'\elem[x,y] \mid I- O_x^*O_y}]_{x,y\in D}
\]
has a feasible solution.
\end{itemize}
\end{lem}

\pfstart
The point (a) merely says that unitaries $\cA(O_x)$ do not change the norm of a vector.

The point (b) follows from~\rf{prp:onegammaFeasible}. 
Indeed, in notation of the this proposition, 
$\Delta_{x,y} = I - O_x^*O_y = 0$ if and only if $x = y$, and 
$e_{x,x} = M\elem[x,x]-M'\elem[x,x]=0$ for all $x\in D$.
\pfend

We can now describe our algorithm for exact state conversion.

\begin{thm}
\label{thm:exactTransformation}
Let $\xi_x\mapsto \tau_x$ be a state conversion problem with pairwise distinct unitary input oracles $O_x$, where $x$ ranges over a finite set $D$.
Assume $(v_x)_{x\in D}$ is a feasible solution to the adversary optimization problem~\rf{eqn:adversary}.
Then, for every $\delta>0$, there exists a quantum algorithm $\cA$ with the following properties:
\begin{itemize}
\itemsep0pt
\item $\cA$ solves state conversion $\xi_x \mapsto \tau_x$ exactly with unidirectional access to $O_x$;
\item for each $x\in D$, we have $\normA|{L_x(\cA) - \Dnorm|v_x|^2 }|\le\delta$.
\end{itemize}
\end{thm}

Together with \rf{thm:mainlower} and \rf{clm:feasibleClosed}, this gives the following main result of the paper:

\begin{thm}
\label{thm:LasVegas}
In the assumptions of \rf{thm:exactTransformation}, the following two sets are equal:
\itemstart
\item the topological closure of the feasible complexity space of the state conversion problem; and
\item the feasible objective space of the corresponding adversary optimisation problem~\rf{eqn:adversary}.
\itemend
\end{thm}

\begin{proof}[Proof of \rf{thm:exactTransformation}]
The idea is as follows.
First, we show that we can transform $G_\xi$ and $G_\tau$ into some $M', M''\succ 0$.
As soon as we get into the space of full-rank Gram matrices, we can use \rf{lem:exactForPositive}.  Namely, we use it to perform the two middle steps in the following chain of transformations:
\begin{equation}
\label{eqn:chainOfTransformations}
G_\xi 
\mapsto (1-\eps)G_\xi + \eps M'
\mapsto (1-\eps)G_\xi + \eps M''
\mapsto (1-\eps)G_\tau + \eps M''
\mapsto G_\tau,
\end{equation}
where $\eps>0$ is some small number.
See \rf{fig:exactAlgorithm}(b).
The main work happens in the third step.
We use Propositions~\ref{prp:workSpaceExtension} and~\ref{prp:scaling} to show that the complexity of the other steps vanishes with $\eps\to 0$.
The final idea is that we perform the last transformation in reverse using \rf{prp:inversion}.

Let us proceed with the proof.
We may assume that $G_\tau \in \cR_\xi$.  Otherwise, neither the state conversion is possible, nor the optimisation problem has a feasible solution.
We may also assume that neither of $\xi_x$ is zero, since then $G_\tau\in\cR_\xi$ implies $\tau_x=0$ and any algorithm always transforms $0\mapsto 0$, so we may drop this input.
%
Consider the matrix $M$ defined by
\begin{equation}
\label{eqn:M}
M\elem[x,y] = \begin{cases}
G_\xi\elem[x,x]=G_\tau\elem[x,x], &\text{if $x = y$;}\\
0,&\text{otherwise.}
\end{cases}
\end{equation}
Clearly, $M\in\cR_\xi$ and $M\succ 0$.
By \rf{lem:cH_xi}(b), the adversary bound~\rf{eqn:adversary} corresponding to the transformation $G_\xi\mapsto M$ has a feasible solution.
Using \rf{cor:main} and continuity of the inner product, we can get Gram matrices achievable from $G_\xi$ that are arbitrarily close to $M$.
Since $M$ is positive definite, there exists a positive definite $M'$ among them.
Let $\cB$ be the algorithm that performs the transformation $G_\xi \mapsto M'$.
 
Using the same argument but with $\xi_x$ replaced by $\tau_x$ and $O_x$ replaced by $O^*_x$, we get $M''\succ 0$ and an algorithm $\cE$ that transforms $G_\tau \mapsto M''$ using the input oracles $O_x^*$.

By \rf{lem:cH_xi}(a), both $M'$ and $M''$ are in $\cR_\xi$.
They are both positive definite.
By point (b) of the same lemma and \rf{lem:exactForPositive}, there exists a quantum algorithm $\cC$ that transforms $M'\mapsto M''$ exactly.
These algorithms are depicted in \rf{fig:exactAlgorithm}(a).

\myfigure{\label{fig:exactAlgorithm}}
{
The algorithms used in the proof of \rf{thm:exactTransformation}.
The Gram matrices $G_\xi$ and $G_\tau$ are not of full rank here, therefore are depicted on the edge of the cone of positive semidefinite matrices.  
The matrix $M$ is positive definite, and so are all the matrices in a sufficiently small circle around $M$.
In (b), the algorithms $\cB_\eps$ and $\cC_\eps$ are the scaled-down versions of $\cB$ and $\cC$.  The algorithm $\cE^{-1}$ is additionally reversed.
As $\eps\to 0$, the algorithm $\cD_\eps$ approaches the line connecting $G_\xi$ and $G_\tau$.
}
{
\[
\begin{tikzpicture}
[
>={Latex[length=2.5mm]},
inner sep=1.5pt,
tochka/.style={circle, fill=black, radius=1pt}
]
\draw (-6,-2) to[out=-15, in=-165] 
node[pos=0.15, tochka, label=below:$G_\xi$] (xi) {} 
node[pos=0.93, tochka, label=below:$G_\tau$] (tau) {} 
(6, -2);
\node at (-7,0) {(a)};
\node [tochka, label=above:$M$] at (0,0) (M) {};
\draw (0,0) circle [radius=2cm] {};
\node [tochka, label=above:$M'$] at (-1, 0.4)  (M1) {};
\node [tochka, label=above:$M''$] at (1, -1)  (M2) {};
\draw [->, decorate, decoration={zigzag, amplitude=.5mm}] (xi) to node[above left] {$\cB$} (M1);
\draw [->] (M1) to node[below left] {$\cC$} (M2);
\draw [->, decorate, decoration={zigzag, amplitude=.5mm}] (tau) to node[above right] {$\cE$} (M2);
\end{tikzpicture}
\]
\bigskip
\[
\begin{tikzpicture}
[
>={Latex[length=2.5mm]},
inner sep=1.5pt,
tochka/.style={circle, fill=black, radius=1pt}
]
\draw (-6,-2) to[out=-15, in=-165] 
node[pos=0.15, tochka, label=below:$G_\xi$] (xi) {} 
node[pos=0.93, tochka, label=below:$G_\tau$] (tau) {} 
(6, -2);
\node at (-7,0) {(b)};
\node [tochka, label=above:$M$] at (0,0) (M) {};
\node [tochka, label=above:$M'$] at (-1, 0.4)  (M1) {};
\node [tochka, label=above:$M''$] at (1, -1)  (M2) {};
\draw [dashed] (xi) to node[tochka, pos=0.3] (N1) {} (M1);
\draw [dashed] (xi) to node[tochka, pos=0.3] (N) {} (M2);
\draw [dashed] (tau) to node[tochka, pos=0.3] (N2) {} (M2);
\draw [->, decorate, decoration={zigzag, amplitude=.4mm}] (xi) to node[above left] {$\cB_\eps$} (N1);
\draw [->] (N1) to node[above right] {$\cC_\eps$} (N);
\draw [->] (N) to node[below] {$\cD_\eps$} (N2);
\draw [->, decorate, decoration={zigzag, amplitude=.4mm}] (N2) to node[above right] {$\cE_\eps^{-1}$} (tau);
\end{tikzpicture}
\]
}

For small enough $\eps>0$, the following table lists the algorithms performing the transformations in~\rf{eqn:chainOfTransformations} with their Las Vegas complexities.
A graphical representation of the algorithms is given in \rf{fig:exactAlgorithm}(b).
\[
\begin{tabular}{ccrc@{$\mapsto\quad$}lcc}
Algorithm &\quad& \multicolumn{3}{c}{Transformation} &\quad& Complexity\\
\hline
$\cB_\eps$ && $G_\xi$&& $(1-\eps)G_\xi + \eps M'$ && $L_x(\cB_\eps) = \eps L_x(\cB)$\\
$\cC_\eps$ && $(1-\eps)G_\xi + \eps M'$ && $(1-\eps)G_\xi + \eps M''$ && $L_x(\cC_\eps) = \eps L_x(\cC)$\\
$\cD_\eps$ && $(1-\eps)G_\xi + \eps M''$ && $(1-\eps)G_\tau + \eps M''$ && $L_x(\cD_\eps) = (1-\eps) \Dnorm|v_x|^2$\\
$\cE^{-1}_\eps$ && $(1-\eps)G_\tau + \eps M''$ && $G_\tau$ && $L_x(\cE_\eps) = \eps L_x(\cE)$\\
\hline
\end{tabular}
\]

The algorithm $\cB_\eps$ is $I\oplus \cB$ of \rf{prp:workSpaceExtension}, where $\cB$ transforms $\eps G_\xi \mapsto \eps M'$.  The complexity follows from \rf{prp:scaling}.
The algorithm $\cC_\eps$ is analogous with $\cC$ transforming $\eps M'\mapsto \eps M''$.
The algorithm $\cE_\eps^{-1}$ is $I\oplus \cE^{-1}$, where $\cE^{-1}$ transforms $\eps M''\mapsto \eps G_\tau$ due to Propositions~\ref{prp:inversion} and~\ref{prp:scaling}.

To get $\cD_\eps$, note that both matrices are positive definite, and their difference is
$(1-\eps)(G_\xi-G_\tau)$, hence, we can use \rf{lem:exactForPositive} with $\sqrt{1-\eps}\, v_x$ as a feasible solution to the corresponding adversary bound~\rf{eqn:adversary}.

The algorithm $\cA$ is the sequential composition of these subroutines, hence, by \rf{prp:sequential}, we have
\[
L_x(\cA) = \eps L_x(\cB) + \eps L_x(\cC) + (1-\eps) \Dnorm|v_x|^2 + \eps L_x(\cE) \to \Dnorm|v_x|^2
\]
as $\eps\to 0$.
\end{proof}

\subsection{Example with Two Labels}
\label{sec:ExampleOfTwoLabels}
Here we consider an example when $D=\{0,1\}$ and the input oracles $O_0\ne O_1$ are unitary.
For normalized states, Gram matrices can be parametrized by a single off-diagonal parameter $a$.  We write
\begin{equation}
G_a = \begin{pmatrix}
1 & a \\
a^* & 1
\end{pmatrix}.
\end{equation}
Consider a transformation $G_a\mapsto G_b$.
In other words, we have that $\ip<\xi_0,\xi_1>=a$ and $\ip<\tau_0,\tau_1>=b$.

\begin{clm}
\label{clm:2lab}
The feasible objective space of the corresponding adversary optimisation problem~\rf{eqn:adversary} is the epigraph of a hyperbola:
\begin{equation}
\label{eqn:2labA}
\sfig{(w_0, w_1) \midB \text{$w_0,w_1\ge 0$ and $\sqrt{w_0w_1} \ge \frac{|a-b|}{\|O_0 - O_1\|}$} }.
\end{equation}
\end{clm}

\pfstart
First, since $O_0$ and $O_1$ are unitaries, we get that
\[
G_a\elem[x,x] - G_b\elem[x,x] = 0 
= \ip<v_x, ((I-O_i^*O_i)\otimes I_\cW)v_x>
\]
for all $v_x$ and $x=0,1$.
So we only have to analyse the off-diagonal term
\[
G_a\elem[0,1]-G_b\elem[0,1] = a - b.
\]

Denote $d = \|I - O_0^*O_1\| = \|O_0 - O_1\|$.
If $v_0, v_1$ is a feasible solution to the adversary optimisation problem, then from~\rf{eqn:advExplicitCondition}, we get
\[
|a - b| = \absA|\ip<v_0, ((I-O_0^*O_1)\otimes I_\cW) v_1>| \le d\,\|v_0\|\cdot \|v_1\|,
\]
implying the lower bound in~\rf{eqn:2labA}.

In the opposite direction, assume $\sqrt{w_0w_1} = |a-b|/d$.
Let $u$ and $v$ be the normalised left and right singular vectors of $I-O_0^*O_1$ with the singular value $d$.
Then, we have
\[
a-b = (\sqrt{w_0}u)^* (I-O^*_0O_1) \frac{(a-b)\sqrt{w_1} v}{|a-b|},
\]
implying that $(w_0,w_1)$ is in the feasible objective space.
The claim follows from~\rf{prp:upwardsClosed}.
\pfend

By \rf{thm:LasVegas}, the topological closure of the feasible complexity space of the corresponding state conversion problem equals~\rf{eqn:2labA}.
We show that, in general, not all points in this set are attained as complexity profiles of the algorithms performing the transformation.

\begin{clm}
Let $O_0=1$ and $O_1=-1$ be 1-dimensional unitaries.
Consider the transformation $G_1\mapsto G_\ii$ (where $\ii$ is the imaginary unit) on these input oracles.
The point $(1/\sqrt2, 1/\sqrt2)$ is in the feasible objective space of the corresponding adversary optimisation problem, but not in the feasible complexity space of the problem.
\end{clm}

\pfstart
The first statement follows from~\rf{clm:2lab}.  It remains to prove there is no algorithm solving the problem with this complexity profile.

Consider an algorithm $\cA$ that performs this transformation in $T$ queries, and assume that it goes through the following Gram matrices during its execution:
\begin{equation}
\label{eqn:Gsequence}
G_1 \mapsto G_{c_1} \mapsto G_{c_2} \mapsto \cdots \mapsto G_{c_{T-1}} \mapsto G_\ii.
\end{equation}

First, we claim that only $G_b$ with $b \in \mathbb{R}$ are achievable from $G_{1}$ in one query.
Indeed, $G_1$ corresponds to a state collection $\xi_0$, $\xi_1$ with $\xi_0 = \xi_1$.
Therefore, the state processed by the input oracle is the same for $x=0$ and $x=1$.
Denote it $\psi'$.
For the states $\tau_0$ and $\tau_1$ after the query, we have 
\[
\ip<\tau_0, \tau_1> = 1 - 2\|\psi'\|^2 \in \bR.
\]

Therefore, among $c_1,\dots,c_{T-1}$ there exists $c_j\in \bR\setminus\{1\}$.
Write the algorithm as a sequential composition $\cA = \cC * \cB$, where $\cB$ performs the transformation $G_1\mapsto G_{c_j}$ in~\rf{eqn:Gsequence}, and $\cC$ the transformation $G_{c_j}\mapsto G_\ii$.

Using~\rf{clm:2lab}, we get that for any point $(w_0,w_1)$ in the feasible objective space for transformation $G_a\mapsto G_b$ with our choice of input oracles
\[
w_0 + w_1 \ge 2\sqrt{w_0w_1} \ge |a-b|.
\]
Hence, by \rf{thm:mainlower}:
\[
L_0(\cB) + L_1(\cB) \ge |1 - c_j|
\qqand
L_0(\cC) + L_1(\cC) \ge |c_j - \ii|.
\]
Combining this with \rf{prp:sequential} and the triangle inequality in $\bC$, we get
\[
L_0(\cA) + L_1(\cA) 
=L_0(\cB) + L_0(\cC) +
L_1(\cB) + L_1(\cC)
\ge |1 - c_j| + |c_j - \ii| > |1 - \ii| = \sqrt2.
\]
Thus, $(1/\sqrt2, 1/\sqrt2)$ is not in the feasible complexity space.
\pfend

Let us now move to the case when $O_0$ and $O_1$ are contractions.
Our goal is to show that Theorems~\ref{thm:exactTransformation} and~\ref{thm:LasVegas} are false in this case.
For that, consider the transformation $G_1\mapsto G_0$ with the input oracles in $\bC^2$ given by
\[
O_0 = \begin{pmatrix}
1&0\\0&0
\end{pmatrix}
\qqand
O_1 = \begin{pmatrix}
0&-1\\0&0
\end{pmatrix}.
\]

\begin{clm}
For the above problem, the adversary optimisation problem has a feasible solution, but there is no algorithm performing the required transformation exactly.
\end{clm}

\pfstart
The feasible solution is $v_0 = \k0$ and $v_1=\k1$.
Let us prove there is no algorithm solving the problem.

The initial Gram matrix $G_1$ means that we have equal initial states $\xi_0 = \xi_1$ of unit norm.
As $G_0\ne G_1$, the algorithm has to make at least one query.
Consider the first query.
Since $\xi_0=\xi_1$, the state given to the oracle is the same for both inputs, denote it $\psi'$.
We may assume $\psi'\ne 0$, as otherwise this query can be ignored.
Let $\psi_x$ be the state of the algorithm after the query on input $x$.

Without loss of generality, the algorithm is sliced, hence, $\psi' = \alpha\k0 + \beta \k1$ for some $\alpha,\beta\in\bC$.
As $\psi'\ne 0$, either $\alpha\ne 0$, or $\beta\ne 0$.
If $\alpha\ne 0$, then $\norm|\psi_1| < 1$.
If $\beta\ne 0$, then $\norm|\psi_0| < 1$.
In either case, it is impossible to get both terminal states $\tau_0$ and $\tau_1$ to have unit norm.
Therefore, there is no algorithm solving the problem.
\pfend

\subsection{Boolean Function Evaluation}
\label{sec:booleanFunction}

Here, we derive the adversary bound for Boolean function evaluation as a simple special case of \rf{thm:LasVegas}.

Let us specify exactly what we mean by Boolean function evaluation in this context.
Let $f\colon D\to\bool$ with $D\subseteq\cube$ be a (partial) Boolean function.
We assume the input oracle $O_x$ encodes $x\in D$ in the phase, and we consider the multi-oracle settings.
That is, there are $n$ unitary input oracles $O^{(i)}_x\colon \bC\to\bC$ defined by $O^{(i)}_x = (-1)^{x_i}$.
Since $O_x$ is Hermitian, there is no difference between unidirectional and bidirectional access to this oracle.
We also assume that the function is evaluated in the phase, exactly and coherently.
That is, the output space $\cK = \bC$ and the goal is to map $\ket|0>\mapsto (-1)^{f(x)}\ket|0>$.
Since the space is one-dimensional, this can also be seen as an instance of unitary implementation.

We get that $\ip<\xi_x,\xi_y> - \ip<\tau_x,\tau_y> = 2\cdot 1_{f(x)\ne f(y)}$, where $1_P$ is the indicator variable.
Similarly, $I - O_x^*O_y = 2\bigoplus_{j=1}^n 1_{x_j\ne y_j}$, where $\bigoplus$ is a direct sum of $n$ matrices, each of size $1\times 1$, resulting in an $n\times n$ diagonal matrix.
Dividing by 2, we get the adversary optimisation problem
\begin{equation}
\label{eqn:BooleanOne}
\onegamma\sB[1_{f(x)\ne f(y)} \midA \bigoplus_{j=1}^n 1_{x_j \ne y_j}].
\end{equation}
It is similar to the corresponding expression in~\cite{belovs:variations}, except that it is  unidirectional.
In \rf{sec:bidirectionality}, we will show that since $O_x = O_x^*$, the unidirectional version is equal to the usual bidirectional relative $\gamma_2$-bound.

The multi-objective optimisation problem~\rf{eqn:BooleanOne} exactly characterises the Las Vegas complexity of each of the $n$ individual input symbols on every input $x\in D$.

It is also possible to substitute the input oracle with the one that encodes $x_j$ in the phase.  
Namely, with the oracle $\tO^{(i)}_x\colon \bC^2\to\bC^2$ given by $\ket|b> \mapsto \ket |b\oplus x_j>$, where $\oplus$ is XOR here.
Indeed, in the Fourier basis, $\tO^{(i)}_x = I_1 \oplus O^{(i)}_x$, hence, the algorithm can just ignore the $I_1$ part.
The same holds for the output, if we require that the algorithm has to perform the transformation $\ket|b>\mapsto \ket|b\oplus f(x)>$ in the output space $\cK = \bC^2$.

\section{Subspace Conversion}
\label{sec:subspaceConversion}

In \rf{sec:linearConsistency}, we continue the settings of \rf{thm:exactTransformation} but without the assumption that input oracles are pairwise distinct.
This leads to our investigation of the linear consistency of feasible solutions.
The corresponding problem can be formulated as subspace conversion, which we analyse in \rf{sec:subspaceConversionProblem}, define the corresponding notion of complexity and extend the connection to the adversary bound.
Finally, in \rf{sec:compositionRevisited}, we revisit the functional composition property of \rf{sec:composition}.
The notion of complexity we introduced for the subspace conversion problem will allow us to formulate and prove a simpler estimate on the complexity of the composed algorithm.

\subsection{Linear Consistency}
\label{sec:linearConsistency}

Throughout the section, we assume we have a state conversion problem $\xi_x\mapsto \tau_x$ in $\cK$ with input oracles $O_x$, where $x$ ranges over $D$.
We will be particularly interested in pairs of inputs $x$, $y$ with $O_x = O_y$.

For $O\in\{O_x\mid x\in D\}$, let $D_O = \{x\in D\mid O_x = O\}$ and $\cK_O = \spn\{\xi_x\mid x\in D_O\}$.
The important point is that the algorithm performs the same linear transformation $\cA(O)$ on all $x\in D_O$.
Therefore, the pairs $\xi_x\mapsto\tau_x$ for $x\in D_O$ should be linearly consistent.
The adversary optimisation problem is in accord with this requirement as shown in the next result.

\begin{prp}
\label{prp:consistencySolution}
Assume that the adversary optimisation problem~\rf{eqn:adversary} for a state conversion problem $\xi_x\mapsto \tau_x$ with contraction oracles $O_x$ has a feasible solution.
Then, for each $O$, there exists a linear transformation $T_O\colon \cK_O\to\cK$ such that $\tau_x = T_O\xi_x$ for all $x\in D_O$.
Moreover, if $O$ is unitary, $T_O$ is unitary.
\end{prp}

\pfstart
Fix $O$, and restrict the optimisation problem to $D_O$.
Since $O$ is a contraction, $I-O^*O$ exists semi-definite, and $S = ((I-O^*O)\otimes I_{\cW})^{1/2}$ is defined.
By~\rf{eqn:advExplicitCondition}, we have
\[
\ip<\xi_x, \xi_y> - \ip<\tau_x,\tau_y> = \ipA<v_x,\;  ((I-O^*O)\otimes I_{\cW}) v_y> = \ipA<S v_x,\;S v_y>
\]
for all $x,y\in D_O$.  Thus, there exists a unitary that maps $\xi_x \mapsto \tau_x \oplus Sv_x$ for all $x$.  Hence, the mapping $T_O\colon \xi_x\mapsto \tau_x$ is linear.  If $O$ is unitary, then $S=0$, and $T_O$ is unitary.
\pfend

Note that the latter result is false for general linear transformations.
For instance, if $O=2I$ it is easy to construct a feasible solution for a non-linear state conversion $0\mapsto \k0$.
This does \emph{not} contradict \rf{thm:mainupper} though, because there the initial state is perturbed.
Effects like this is the main reason why we focus on contraction oracles.

Let us also consider linear consistency of feasible solutions.

\begin{defn}[Linear consistency of feasible solutions]
\label{defn:consistency}
We say that a feasible solution $v_x$ to the adversary optimisation problem~\rf{eqn:adversary} is \emph{linearly consistent} if, for each $O$, there exists a linear transformation $V_O\colon \cK_O\to \cM\otimes\cW$ such that $v_x = V_O\xi_x$ for all $x\in D_x$.
\end{defn}

One way to ensure this condition is to impose the following.

\begin{defn}[Linear independence assumption]
\label{defn:independence}
We say that a state conversion problem satisfies the \emph{linear independence assumption} if, for each $O$, the vectors in $\{\xi_x\mid x\in D_O\}$ are linearly independent.
\end{defn}

Under this assumption, we are losing nothing in relation to the transformation performed.
For each $O$, we can uniquely extend this state conversion to all $\xi\in\cK_O$ by linearity.
We will call the latter the \emph{linearly extended state conversion problem}.

\begin{prp}
\label{prp:linearConsistency}
We have
\begin{itemize}
\item[(a)] Any feasible solution obtained via \rf{thm:mainlower} is linearly consistent.
\item[(b)] Linear independence assumption implies linear consistency for all feasible solutions.  
\item[(c)] Moreover, under linear independence assumption, any feasible solution can be uniquely extended to a linearly consistent feasible solution to the linearly extended state conversion problem.
\end{itemize}
\end{prp}

\pfstart
For Point (a), use $V_O = \TotalQuery(\cA, O)$ from~\rf{eqn:TotalQuery}.
Point (b) is obvious.
Point (c) follows by linearity.
Let $x_i$ range over $D_O$ and $y_j$ over $D_{O'}$.  If
\[
\ip<\xi_{x_i}, \xi_{y_j}> - \ip<\tau_{x_i}, \tau_{y_j}> = \ipA<v_{x_i},\;  ((I-O^*O')\otimes I_{\cW}) v_{y_j}>
\]
for all $x_i$ and $y_j$, then
\[
\ipB<\sum_{i} a_i\xi_{x_i}, \sum_{j} b_j\xi_{y_j}> - \ip<\sum_{i} a_i\tau_{x_i}, \sum_{j} b_j\tau_{y_j}> = \ipB<\sum_{i} a_i v_{x_i},\;  ((I-O^*O')\otimes I_{\cW}) \sum_{j} b_jv_{y_j}>
\]
for all complex $a_i$ and $b_j$.
\pfend

Contrary to Propositions~\ref{prp:consistencySolution} and~\ref{prp:linearConsistency}(a), feasible solutions to the adversary optimisation problem need not satisfy linear consistency.
For example, let $D = \{0,1,+,-\}$, $\cK$ be a qubit, $\xi_0 = \k0$, $\xi_1 = \k1$, $\xi_+ = (\k0+\k1)/\sqrt2$, and $\xi_- = (\k0-\k1)/\sqrt2$.
We let $O_x = I$ and $\tau_x = \xi_x$ for all $x\in D$.
This problem in trivially solvable in 0 queries.
But the following is a feasible solution to~\rf{eqn:adversary}, which is not linearly consistent:
$v_0 = v_1 = 0$, and $v_+ = v_- = \k0$.

This poses a problem for strengthening \rf{thm:LasVegas}.
The feasible solution above gives us a point $(0,0,1,1)$ in the feasible objective space.
On the other hand, by the parallelogram identity, \rf{prp:parallelogram}, we have that for every algorithm $\cA$:
\[
L_+(\cA) + L_-(\cA) = L_0(\cA) + L_1(\cA),
\]
implying that the point $(0,0,1,1)$ is not in the topological closure of the feasible complexity space.

One can say that this example is artificial.
There is no need to deteriorate the solution $v_0=v_1=v_+=v_-=0$ by increasing $v_+$ and $v_-$.
The following result states that this is a general observation.
Recall that a solution to a multi-objective optimisation problem is called \emph{Pareto optimal} if it is not strictly dominated by any other solution.
In our case that means that there is no other feasible solution $v_x'$ to the same optimisation problem such that $\Dnorm|v_x'|^2\le \Dnorm|v_x|^2$ for all $x$ and $\Dnorm|v_x'|^2 < \Dnorm|v_x|^2$ for some $x$.

\begin{prp}
Any Pareto optimal solution to the adversary optimisation problem with contraction input oracles is linearly consistent.
\end{prp}

\begin{proof}
Let $v_x$ be a feasible solution.
Take any $O$.  We will show that either $v_x$ is linear consistent on $D_O$, or there is a feasible solution that strictly dominates $v_x$.

By an argument like in \rf{prp:linearConsistency}(b) and (c), there exists a feasible solution $v'_x$ that is linearly consistent on $D_O$ and equal to $v_x$ outside of $D_O$.
That is, there exists a linear map $V'\colon \cK_O\to\cM\otimes\cW$ such that $v'_x = V'\xi_x$ for all $x\in D_O$.
Recall the conditions~\rf{eqn:advExplicitCondition}:
\[
\ip<\xi_x, \xi_y> - \ip<\tau_x,\tau_y> = \ipA<v_x,\;  ((I-O^*_xO_y)\otimes I_{\cW}) v_y> .
\]
First, let us consider these constraints for $x\in D_O$ and $y\notin D_O$.
(The constraints with $x\notin D_O$ and $y\in D_O$ are equivalent to these ones due to the symmetry imposed on a unidirectional relative $\gamma_2$-optimisation problem, \rf{defn:onegamma}).
Let $\Pi_1$ denote the projector onto the span of $((I-OO_y)\otimes I_{\cW}) v_y$ as $y$ ranges over $D\setminus D_O$.
The constraints~\rf{eqn:advExplicitCondition} are linear in $v_x$ and define $\Pi_1 v_x$ uniquely.
Therefore, $\Pi_1 v_x = \Pi_1 v'_x$, and the mapping $\xi_x \mapsto \Pi_1 v_x$ is linear for $x\in D_O$.

Next, consider the constraints~\rf{eqn:advExplicitCondition} for $x,y\in D_O$.
Similarly to the proof of \rf{prp:consistencySolution}, the operator $(I-O^*O)\otimes I_W$ is positive semi-definite.
Let $S = ((I-O^*O)\otimes I_W)^{1/2}$ and $\Pi_2$ denote the projector onto its range.
We claim that the mapping $\xi_x \mapsto \Pi_2 v_x$ is linear on $D_O$ as well.
Indeed, for $x,y\in D_O$, we have:
\[
\ip<S v'_x, S v'_y> = \ip<\xi_x,\xi_y> - \ip<\tau_x,\tau_y> = \ip<S\Pi_2 v_x, S\Pi_2 v_y>
\]
Hence, there is a unitary $U$ that maps $Sv'_x \mapsto S\Pi_2 v_x$.  
Then, $\Pi_2 v_x = S^+USV'\xi_x$, where $S^+$ is the Moore-Penrose pseudoinverse.

Let $\Pi$ denote the projector onto the span of $\Pi_1$ and $\Pi_2$.
Since both $\xi_x \mapsto \Pi_1v_x$ and $\xi_x\mapsto \Pi_2v_x$ are linear on $D_O$, the same holds for $\xi_x \mapsto \Pi v_x$.
If we replace $v_x$ by $\Pi v_x$ for each $x\in D_O$, we get a feasible solution that is linearly consistent on $D_O$ and dominates $v_x$.
\end{proof}

Thus, by imposing the linear consistency condition of \rf{defn:consistency} we are only losing solutions where some of the objectives $\Dnorm|v_x|^2$ are artificially inflated.
In the following, we will only consider linearly consistent solutions.
By \rf{prp:linearConsistency}(c), we may assume the problem satisfies the linear independence condition.
We have the following generalisation of \rf{thm:exactTransformation}.

\begin{thm}
\label{thm:exactConsistent}
Let $\xi_x\mapsto \tau_x$ be a state conversion problem with unitary oracles $O_x$, where $x$ ranges over a finite set $D$, and which satisfies the linear independence condition.
Assume $(v_x)_{x\in D}$ is a feasible solution to the adversary optimization problem~\rf{eqn:adversary}, and
let $T_O$ and $V_O$ be like in \rf{prp:consistencySolution} and \rf{defn:consistency}.
Then, for every $\delta>0$, there exists a quantum algorithm $\cA$ with the following properties:
\begin{itemize}
\itemsep0pt
\item for every $O$, and $\xi\in\cK_O$, $\cA$ transforms $\xi \mapsto T_O\xi$ on input oracle $O$;
\item moreover, $\normA|{L(\cA, O, \xi) - \Dnorm|V_O\xi|^2 }|\le\delta$.
\end{itemize}
\end{thm}

\pfstart
The proof is a modification of the proof of \rf{thm:exactTransformation}.
Redefine $\cR_\xi$ as the real affine space of $D\times D$ Hermitian matrices $A$ satisfying $A\elem[x,y] = \ip<\xi_x, \xi_y>$ for all $x,y\in D$ with $O_x = O_y$.
The two points of \rf{lem:cH_xi} still hold.
The matrix $M$ is defined by
\begin{equation}
\label{eqn:M2}
M\elem[x,y] = \begin{cases}
G_\xi\elem[x,y]=G_\tau\elem[x,y], &\text{if $O_x = O_y$;}\\
0,&\text{otherwise.}
\end{cases}
\end{equation}
Again, $M\in\cR_\xi$, as well as $M\succ 0$, since it is a block-diagonal matrix whose blocks are Gram matrices of linearly independent collections of vectors.
Other than that, the algorithm is exactly the same as in \rf{thm:exactTransformation}.

The algorithm transforms $\xi_x\mapsto \tau_x$ on the input oracle $O$ for all $x\in D_O$.
By linearity, it transforms $\xi\mapsto T_O\xi$ for all $\xi\in\cK_O$.
The same linearity property holds for all queries made by the algorithm.
Therefore, the Las Vegas query complexity of the $\cD_\eps$ subroutine of the algorithm when the $\cA$ is executed on the initial state $\xi\in\cK_O$ is $(1-\eps)\Dnorm|V_O\xi|^2$.
The complexities of other subroutines tend to zero as $\eps\to 0$, which gives the required result.
\pfend

\subsection{Subspace Conversion Problem}
\label{sec:subspaceConversionProblem}

As one can see, what \rf{thm:exactConsistent} actually solves is the subspace conversion problem from \rf{defn:subspaceConversion}.
Let us restate the problem assuming general input oracles.

\begin{defn}[Subspace Conversion, Restated]
\label{defn:subspaceConversionRestated}
Let $D$ be a set of labels, and $\cM$ and $\cK$ be vector spaces.
For each $x\in D$, let $O_x\colon \cM\to\cM$ be a linear transformation.
A \emph{subspace conversion problem} is given by a collection of linear maps $T_x\colon \cK_x\to \cK$ with $\cK_x\subseteq \cK$, where $x$ ranges over $D$.
Assume that $\cK$ is embedded in the space $\cH$ of a quantum algorithm $\cA$.
We say that the algorithm $\cA$ \emph{solves} the subspace conversion problem $T_x$ on input oracles $O_x$, if, for every $x\in D$, the map $\cA(O_x)$ agrees with $T_x$ on $\cK_x$.
We define the complexity $L_x(\cA)$ on the input $x$ as the supremum of $L(\cA, O, \xi)$ as $\xi$ ranges over the unit vectors in $\cK_x$.
The remaining complexity-related definitions are as in \rf{defn:LAStateConversion}.
\end{defn}

In the case of a single input oracle, the supremum in the definition of $L_x(\cA)$ is just the maximum.
In the case of multiple input oracles of \rf{sec:multipleOracles}, the supremum is understood with respect to the dominance relation $u\le v$.
In other words, it is the entry-wise maximum.  
Note that it is not true, in general, that there exists a unit $\xi\in\cK_x$ such that $L_x(\cA) = L(\cA, O_x, \xi)$, since different input oracles can attain their maxima at different $\xi$.

The corresponding adversary optimisation problem is as follows:
\begin{defn}[Adversary for Subspace Conversion]
Consider a subspace conversion problem as defined above with input oracles $O_x$, and let $K_x$ be the projector onto $\cK_x$.
The corresponding adversary optimisation problem is given by
\begin{equation}
\label{eqn:adversarySubspaceConversion}
\onegamma \sA[K_x^*K_y - T_x^*T_y \mid I-O_x^*O_y  ]_{x,y\in D}.
\end{equation}
\end{defn}

Note that while~\rf{eqn:onegammaMultiobjective} states that $V_x\colon \cK\to\cM\otimes\cW$, in this optimisation problem we actually have $V_x\colon \cK_x \to \cM\otimes\cW$, as $T_x$ is defined on $\cK_x$, and the coimage of $K_x$ is $\cK_x$ as well.
In particular, for the state conversion problem, where each $\cK_x$ is one-dimensional, we get back the definition from~\rf{eqn:adversary}.
On the other extreme, for the unitary implementation problem, where $\cK_x=\cK$ for all $x$, Eq.~\rf{eqn:adversarySubspaceConversion} reads as
\[
\onegamma \sA[I - T_x^*T_y \mid I-O_x^*O_y  ]_{x,y\in D}.
\]

\rf{thm:exactConsistent} can be reformulated as follows.
\begin{thm}
Let $T_x\colon \cK_x\to \cK$ be an isometric subspace conversion problem with unitary input oracles $O_x$ and finite set of labels $D$.
Then, the topological closure of the feasible complexity space of this problem coincides with the feasible objective space of the corresponding adversary optimisation problem~\rf{eqn:adversarySubspaceConversion}.
\end{thm}

\pfstart
If $V_x$ is a feasible solution, then by \rf{eqn:onegammaMultiobjectiveCondition}:
\[
K_x^*K_y - T_x^*T_y = V_x^*((I-O_x^*O_y)\otimes I_\cW)V_y.
\]
Multiplying by $\xi^*$ on the left and $\xi'$ on the right gives
\begin{equation}
\label{eqn:Kxxi}
\ip<K_x\xi, K_y\xi'> - \ip<T_x\xi, T_x\xi'> = \ipA<V_x\xi,\;((I-O_x^*O_y)\otimes I_\cW)V_y\xi'>.
\end{equation}
Hence, $V_x\xi$ is a feasible solution to the adversary optimisation problem of state conversion $\xi\mapsto T_x\xi$ with oracles $O_x$ as $x$ ranges over $D$ and $\xi$ over $\cK_x$.
The theorem follows from~\rf{thm:exactConsistent}.

In other words, Eq.~\rf{eqn:adversarySubspaceConversion} is just a way to write down a linearly consistent feasible solution to an adversary optimisation problem, where $V_x$ acts like $V_O$ in \rf{defn:consistency}.
The objective $\Dnorm|V_x|^2$ is then the largest (entry-wise) complexity on the oracle $O_x$ as $\xi$ ranges over unit vectors in $\cK_x$.
\pfend

\subsection{Composition, Revisited}
\label{sec:compositionRevisited}
Here we revisit the composition properties from \rf{sec:composition} using the notions from \rf{sec:subspaceConversionProblem}.

Let $T_x\colon \cN_x\to\cN$ be a subspace conversion problem as $x$ ranges over $D$, and $\cB$ be an algorithm solving this problem on input oracles $O_x\colon \cM\to\cM$.
For each $x\in D$, let $O'_x\colon\cN\to\cN$ agree with $T_x$ on $\cN_x$.
Let $\xi_x\mapsto\tau_x$ be a state conversion problem with the input oracles $O'_x$.
Assume that a sliced algorithm $\cA$ solves this problem and has the following property:
\begin{equation}
\label{eqn:compositionAssumption}
\forall x\in D\; \forall t\colon \Query_t(\cA, O'_x)\xi_x \in \cN_x.
\end{equation}

We can estimate the complexity of the composed algorithm as the product of complexities of its constituents.

\begin{prp}
\label{prp:compositionRevisited}
Under the above assumption, the composed algorithm $\cA\circ\cB$ defined in \rf{prp:composition} solves the state conversion problem $\xi_x\mapsto\tau_x$ with input oracles $O_x$.
Moreover,
\begin{equation}
\label{eqn:compositionRevisited}
L_x(\cA\circ\cB) \le L_x(\cA)L_x(\cB).
\end{equation}
\end{prp}

\begin{proof}
By~\rf{eqn:compositionAssumption} and the fact that $\cB$ solves the subspace conversion problem, we have that $O'_x$ and $\cB(O_x)$ satisfy the condition~\rf{eqn:inputOracleConsistency} of \rf{prp:inputSpaceExtension}.
Therefore, we can replace the input oracle $O'_x$ of $\cA$ by $\cB(O_x)$.
The first statement then follows from \rf{prp:composition}.

Concerning complexity, we have:
\begin{align*}
L(\cA\circ\cB, O_x,\xi_x) &= \sum_t L\sB[\cB, O_x, {\Query_t\sA[\cA, \cB(O_x)]}\xi_x] \\
& = \sum_t L\sB[\cB, O_x, {\Query_t\sA[\cA, O'_x]}\xi_x] \\
& \le L_x(\cB)\sum_{t} \norm|\Query_t\sA[\cA, O'_x]|^2 = L_x(\cB) L_x(\cA).
\end{align*}
Here, we used~\rf{eqn:composition} on the first step,
\rf{prp:inputSpaceExtension} on the second step,
and the definition of $L_x(\cB)$ and \rf{prp:scaling} on the third step.
\end{proof}

In the above proposition it is assumed that the algorithm $\cA$ has a single input oracle (while $\cB$ can have multiple input oracles).
Similarly, it is possible to get an analogue of~\rf{eqn:compositionMultipleOracles}.
Again, assume $\cA$ has multiple input oracles $O^{(i)}\colon \cN^{(i)}\to\cN^{(i)}$.
For each $i$, let $T^{(i)}_x\colon \cN^{(i)}_x \to \cN^{(i)}$ be a subspace conversion problem with $x$ ranging over $D$, 
and $\cB^{(i)}$ be an algorithm that solves the above problem with input oracle $O_x\colon \cM\to\cM$.
By \rf{prp:directSum}, the algorithm $\cB = \bigoplus_i \cB^{(i)}$ solves state conversion $T_x = \bigoplus_i T^{(i)}_x \colon \cN_x \to \cN$, where $\cN_x = \bigoplus_i \cN^{(i)}_x$ and $\cN = \bigoplus_i \cN^{(i)}$.
Let, for each $x\in D$ and $i$, ${O'_x}^{(i)}$ be a linear map on $\cN^{(i)}$ that agrees with $T^{(i)}_x$ on $\cN^{(i)}_x$, and denote $O'_x = \bigoplus_i {O'_x}^{(i)}$.
Then, using a similar estimate as in the proof of \rf{prp:compositionRevisited}, but with~\rf{eqn:compositionMultipleOracles} instead of~\rf{eqn:composition}, we get:
\[
L_x(\cA\circ \cB) \le \sum_{i} L_x(\cA)\elem[i]\cdot L_x(\cB^{(i)}).
\]
Above we assumed for simplicity that all the subspace conversion problems $T^{(i)}_x$ have the same set of labels $D$.
This is without loss of generality.
If the $i$-th problem has the set of labels $D^{(i)}$, it is possible to take $D$ as the Cartesian product $D = \prod_i D^{(i)}$ or some subset thereof.

\section{Bidirectionality}
\label{sec:bidirectionality}
In this section, we consider aspects specific to bidirectional access to the input oracle.
In particular, we show how one can obtain the main results from~\cite{belovs:variations}.
As mentioned in the introduction, bidirectional case is just a special case of the unidirectional case.
\begin{prp}
\label{prp:bidirectional}
For each state conversion problem with unitary input oracle $(O_x)_{x\in D}$, the feasible complexity space assuming bidirectional access to the oracle $O_x$ coincides to the feasible complexity space assuming unidirectional access to the oracle $O_x\oplus O_x^*$.
Moreover, the corresponding Monte Carlo complexities differ at most by a factor of 2.
\end{prp}

\pfstart
Each algorithm $\cA$ with bidirectional access to $O_x$ can be simulated with unidirectional access to $O_x\oplus O_x^*$ by using the parts $O_x$ and $O_x^*$ of the oracle to process direct and reverse queries of $\cA$.  Both Monte Carlo and Las Vegas complexities do not change.

On the other hand, if $\cA$ has unidirectional access to $O_x\oplus O_x^*$, it can be simulated with bidirectional access to $O_x$ by first processing the $O_x$-part with the direct query, and then the $O_x^*$-part with the reverse query.  Las Vegas complexity does not change, and the Monte Carlo complexity grows by a factor of 2.
\pfend

Let us define the (bidirectional) relative $\gamma_2$-norm.
We start with the single-objective version, which is the version used in~\cite{belovs:variations}.

\begin{defn}[Bidirectional relative $\gamma_2$-bound]
\label{defn:twogamma}
Let $\cK$, and $\cM$ be vector spaces, and $D$ be a set of labels.
Let $E = \{E_{xy}\}$ and $\Delta = \{\Delta_{xy}\}$, where $x,y\in D$ be two families of linear operators: $A_{xy}\colon \cK\to\cK$ and $\Delta_{xy}\colon \cM\to\cM$.

The \emph{relative $\gamma_2$-norm}
\[
\twogamma(E | \Delta) = \twogamma(E_{xy} \mid \Delta_{xy})_{x,y\in D},
\]
is defined as the optimal value of the following optimisation problem, where $U_x$ and $V_x$ are linear operators,
\begin{subequations}
\label{eqn:twogammaMultidimensional}
\begin{alignat}{3}
&\mbox{\rm minimise} &\quad& \max\nolimits_{x\in D} \max\{\norm|U_x|^2, \norm|V_x|^2\} &\quad&\\
& \mbox{\rm subject to}&&  
E_{xy} = U_x^* (\Delta_{xy}\otimes I_{\cW}) V_y && \text{\rm for all $x, y\in D$;}  \\
&&& \text{$\cW$ is a vector space}, &&
U_x, V_x \colon \cK\to \cM\otimes\cW.
\end{alignat}
\end{subequations}
\end{defn}

The one-dimensional version is :
\begin{subequations}
\label{eqn:twogamma}
\begin{alignat}{3}
&\mbox{\rm minimise} &\quad& \max\nolimits_{x\in D} \max\{\norm|u_x|^2, \norm|v_x|^2\} &\quad&\\
& \mbox{\rm subject to}&&  
e_{xy} = \ipA<u_x,\;  (\Delta_{xy}\otimes I_{\cW}) v_y> && \text{\rm for all $x, y\in D$;}  \\
&&& \text{$\cW$ is a vector space}, &&
u_x, v_x \in \cM\otimes\cW.
\end{alignat}
\end{subequations}

The relative $\gamma_2$-norm can be also defined in terms of the unidirectional $\gamma_2$-bound as
\begin{equation}
\label{eqn:bidirectional}
\twogamma(E_{xy} \mid \Delta_{xy})_{x,y\in D} = \onegamma (\widetilde E_{xy} \mid \widetilde\Delta_{xy})_{x,y\in D\cup D'}.
\end{equation}
Here $D' = \{x'\mid x\in D\}$ is a disjoint copy of $D$,
and $\widetilde E$ and $\widetilde \Delta$ are defined as
\begin{align*}
\tE_{x,y'} &= \tE_{x',y}^* = E_{x,y},&
\tE_{x,y}&=\tE_{x',y'}=0,\\
\tDelta_{x,y'} &= \tDelta_{x',y}^* = \Delta_{x,y},&
\tDelta_{x,y}&=\tDelta_{x',y'}=0
\end{align*}
for all $x,y\in D$.
This instantly gives a dual for the the one-dimensional version of the bound, which was already proven in~\cite{belovs:variations}:

\begin{thm}
\label{thm:twogammadual}
The optimal value of~\rf{eqn:twogamma} is equal to the optimal value of the following optimization problem:
\begin{subequations}
\label{eqn:dualAdversary}
\begin{alignat}{2}
&\mbox{\rm maximise} &\quad& \|\Gamma\circ E\| \\
& \mbox{\rm subject to}&&  \|\Gamma\circ \Delta\|\le 1, 
\end{alignat}
\end{subequations}
where $\Gamma$ ranges over $D\times D$ matrices.
\end{thm}

\pfstart
Use the above representation, \rf{thm:onegammadual}, and the fact that
\[
\|A\| = \lambda_{\max} 
\begin{pmatrix}
0 & A \\ A^* & 0
\end{pmatrix}
.\qedhere
\]
\pfend

Now let us move to the connection between unidirectional and bidirectional oracles.
By \rf{prp:bidirectional}, unidirectional access to $O_x$ is equivalent to bidirectional access to $O_x\oplus O_x^*$.
The following proposition shows that we can substitute unidirectional $\gamma_2$ bound with oracle $O_x\oplus O^*_x$ with bidirectional $\gamma_2$-norm with oracle $O_x$.

\begin{prp}
\label{prp:1to2}
Let $O_x = O_x^{(1)}\oplus \cdots\oplus O_x^{(s)}$ be unitary oracles as $x$ ranges over $D$, and assume that $e_{y,x} = e^*_{x,y}$ and $e_{x,x}=0$ are complex numbers for all $x,y\in D$.
Consider the following two optimization problems 
\[
\twogamma(e_{x,y} \mid I-O_x^*O_y)_{x,y\in D} 
\qqand
\onegamma\sA[e_{x,y} \mid I-(O_x^*O_y \oplus O_xO_y^*)]_{x,y\in D} .
\]
Then, for every collection $(L_x)_{x\in D}$, with $L_x\in \bR^s$, the following statements are equivalent:
\itemstart
\item[(a)] there exists a feasible solution $u_x, v_x$ to the first optimization problem with $L_x = (\Dnorm|u_x|^2 + \Dnorm|v_x|^2)/2$;
\item[(b)] there exists a feasible solution $u_x, v_x$ to the first optimization problem with $L_x = \Dnorm|u_x|^2 = \Dnorm|v_x|^2$;
\item[(c)] there exists a feasible solution $\tilde v_x$ to the second optimization problem with $L_x = \Dnorm|\tilde v_x|^2$.
\itemend
\end{prp}

\pfstart
First, let us prove $(a)\Rightarrow(c)$.
Assume we have a feasible solution $u_x$, $v_x$ to~\rf{eqn:twogamma} with $\Delta_{x,y} = I - O_x^*O_y$.  Let us denote $\tO_x = O_x\otimes I_{\cW}$.
In particular, we have
\[
e_{x,y} = \ip<u_x, (I-\tO_x^*\tO_y) v_y>
\;\Longrightarrow\;
e_{x,y} = \ip<u_x, v_y> - \ip<\tO_x u_x, \tO_y v_y> ,
\]
and
\[
e_{y,x} = \ip<u_y, (I - \tO_y^*\tO_x) v_x>
\;\Longrightarrow\;
e_{x,y} = \ip<v_x, u_y> - \ip<\tO_x v_x, \tO_y u_y>.
\]

Consider the following two equalities
\begin{align*}
\Bigl<u_x + v_x,\; \rlap{$\sA[I - \tO_x^*\tO_y] (u_y + v_y)\Bigr>\hfil$} \\
&&= &&
\ip<u_x, u_y> &&+&& \ip<u_x, v_y> &&+&& \ip<v_x, u_y> &&+&& \ip<v_x, v_y>\\
&&-&&\ip<\tO_xu_x, \tO_y u_y> &&-&& \ip<\tO_xu_x, \tO_yv_y> &&-&& \ip<\tO_xv_x, \tO_y u_y> &&-&& \ip<\tO_xv_x, \tO_yv_y>,
\end{align*}
and
\begin{align*}
\Bigl<\tO_xu_x-\rlap{$\tO_xv_x,\; \sA[I-\tO_x\tO_y^*] (\tO_yu_y-\tO_yv_y)\Bigr>\hfil$} \\
&&= &&\ip<\tO_xu_x, \tO_yu_y> &&-&& \ip<\tO_xu_x, \tO_yv_y> &&-&& \ip<\tO_xv_x, \tO_yu_y> &&+&& \ip< \tO_xv_x, \tO_yv_y>\\
&&-&&\ip<u_x, u_y> &&+&& \ip<u_x, v_y> &&+&& \ip<v_x, u_y> &&-&& \ip<v_x, v_y>.
\end{align*}

Hence,
\[
\tilde v_{x} = \frac{(u_x+v_x) \oxplus (\tO_x u_x - \tO_x v_x)}2
\]
is a feasible solution to~\rf{eqn:onegamma} with $\Delta_{x,y} = I-(O_x^*O_y \oplus O_xO_y^*)$.
We have
\[
\Dnorm|\tilde v_x|^2 
= \frac{\DnormA|u_x+v_x|^2 + \DnormA|\tO_x u_x - \tO_x v_x|^2} 4
= \frac{\DnormA|u_x+v_x|^2 + \DnormA|u_x - v_x|^2} 4
= \frac{\DnormA|u_x|^2 + \DnormA|v_x|^2} 2,
\]
where we used~\rf{eqn:DnormScaling}, \rf{eqn:DnormSum}, \rf{eqn:DnormUnitary} and~\rf{eqn:DnormParallelogram}, respectively.
This proves that $(a)\Rightarrow (c)$.

Now let us prove $(c)\Rightarrow (b)$.
Assume $\tilde v_x$ is a feasible solution to the second optimization problem.
Let $\tilde v'_x$ and $\tilde v''_x$ be the parts of $\tilde v_x$ processed by $I-O_x^*O_y$ and $I-O_x O_y^*$, respectively.
Set
\[
u_x = \tilde v'_x \oxplus \tO_x^* \tilde v''_x
\qqand
v_x = \tilde v'_x \oxplus [-\tO_x^* \tilde v''_x].
\]
Then,
\begin{align*}
\ip<u_x, ((I-O_x^*O_y)\otimes (I_\cW\oplus I_\cW))v_y> &=
\ip<\tilde v'_x, (I-\tO_x^*\tO_y) \tilde v'_y> 
+ \ip<\tO_x^* \tilde v''_x, (\tO_x^*\tO_y - I) \tO_y^*\tilde v''_y> \\
&=\ip<\tilde v'_x, (I-\tO_x^*\tO_y ) \tilde v'_y> 
+ \ip<\tilde v''_x, (I-\tO_x\tO_y^* ) \tilde v''_y>  = e_{x,y}.
\end{align*}
Again, using~\rf{eqn:DnormSum} and~\rf{eqn:DnormUnitary}, $\Dnorm|u_x| = \Dnorm|v_x| = \Dnorm|\tilde v_x|$.
This proves $(c)\Rightarrow (b)$.  The remaining implication $(b)\Rightarrow (a)$ is obvious.
\pfend

Therefore, we can make the following definition.
\begin{defn}
The multi-objective bidirectional relative $\gamma_2$-optimisation problem
\[
\twogamma \sA[e_{x,y}\mid \Delta_{x,y}]_{x,y\in D}
\]
is defined as
\begin{subequations}
\label{eqn:twogammaMultiobjective}
\begin{alignat}{3}
&\mbox{\rm minimise} &\quad& \sB[\frac{\Dnorm|u_x|^2 + \Dnorm|v_x|^2}2]_{x\in D} &\quad& \label{eqn:twogammaMultiobjectiveA}\\
& \mbox{\rm subject to}&&  
e_{xy} = \ipA<u_x,\;  (\Delta_{xy}\otimes I_{\cW}) v_y> && \text{\rm for all $x, y\in D$;}  \\
&&& \text{$\cW$ is a vector space}, &&
u_x, v_x \in \cM\otimes\cW.
\end{alignat}
\end{subequations}
Alternatively, one may substitute~\rf{eqn:twogammaMultiobjectiveA} with
\[
\mbox{\rm minimise} \quad \sB[\max\sfigA{\Dnorm|u_x|^2,\; \Dnorm|v_x|^2}]_{x\in D}.
\]
\end{defn}

\begin{defn}[Bidirectional Adversary Optimisation Problem]
Assume $\xi_x\mapsto \tau_x$ is a state conversion problem with bidirectional input oracles $O_x\colon \cM\to\cM$, as $x\in D$.
Its \emph{adversary optimisation problem} is
\begin{equation}
\label{eqn:twoAdversary}
\twogamma\sB[\ip<\xi_x, \xi_y> - \ip<\tau_x,\tau_y> \mid I_\cM - O_x^*O_y]_{x,y\in D}.
\end{equation}
\end{defn}

An important corollary is as follows.

\begin{cor}
Assuming bidirectional access, the adversary bound~\rf{eqn:adversary} can be replaced with the corresponding bidirectional version~\rf{eqn:twoAdversary}, and the results of the corresponding Theorems~\ref{thm:mainlower}, \ref{thm:mainupper}, \ref{thm:exactTransformation}, \ref{thm:LasVegas}, and \rf{cor:main}
still hold.
\end{cor}

For instance, \rf{cor:main} after this transformation is the main technical result from~\cite{belovs:variations} with slightly better dependence on $\eps$.
And the adversary bound for Boolean function evaluation~\rf{eqn:BooleanOne} equals
\begin{equation}
\label{eqn:BooleanTwo}
\twogamma\sB[1_{f(x)\ne f(y)} \midA \bigoplus_{j=1}^n 1_{x_j \ne y_j}],
\end{equation}
which is equivalent to the known bound from~\cite{reichardt:spanPrograms}.
The corresponding dual~\rf{eqn:dualAdversary} is the lower bound from~\cite{hoyer:advNegative}.

\section{Unitary Permutation Inversion}
\label{sec:permutationInversion}
The goal of this section is to prove a separation between unidirectional and bidirectional access to an oracle on a natural problem.
We will achieve this using the following problem.

\begin{defn}[Unitary Permutation Inversion]
The set of labels is the set of permutations on $n$ elements $D=\mathfrak{S}_n$.
For each $\pi\in \mathfrak S_n$, let $O_\pi\colon \bC^n\to\bC^n$ be the input oracle defined by $O_\pi\k i = \ket|\pi(i)>$ for all $i\in[n]$
The task is to find $\pi^{-1}(1)$.
\end{defn}

First note that this problem is different from the usual permutation inversion problem.
In the latter, the permutation $\pi$ is encoded using the standard input oracle $\k i\k b\mapsto \k i\ket|b\oplus \pi(i)>$.
The latter is a well-known problem, first defined in~\cite{bennett:strengths}.  It is similar to Grover's search, but different enough to complicate direct reductions from the lower bound for unstructured search.
Ambainis~\cite{ambainis:adv} gave a tight lower bound of $\Omega(\sqrt n)$.
Nayak~\cite{nayak:invertingAPermutation} gave a direct reduction from unstructured search.
See also a recent paper by Rosmanis~\cite{rosmanis:invertingPermutation}.

Since the unidirectional and bidirectional access are equivalent for standard oracle, we resort to the unitary oracle.
The reason of requiring $\pi$ to be a permutation is solely to ensure that $O_\pi$ is a unitary.

The problem can be trivially solved in one query with bidirectional access:
apply $O^*_\pi$ to $\ket|1>$ and read out the result.
Since unitary inversion using the standard oracle requires $\Omega(\sqrt n)$ queries, this means that the unitary permutation oracle $\k i\mapsto \ket|\pi(i)>$ cannot be simulated by the standard oracle.

Intuitively, it seems the problem should be hard for unidirectional input oracles.
We show that this is indeed the case.

\begin{thm}
Any quantum query algorithm solving the unitary inversion problem (with bounded error and non-coherently) with unidirectional access to the input oracles has to make $\Omega(\sqrt n)$ queries.
\end{thm}

Note, however, that there is no matching upper bound.
Grover's search cannot be directly applied here because of the unidirectional access.
The remaining part of this section is devoted to the proof of this theorem.
The proof relies on \rf{thm:weakduality}, and we have to find the adversary matrix $\Gamma$ from~\rf{eqn:onegammadual}.

Interestingly, the analysis is a variant of the usual positive-weighted adversary, but it 
is different from the one used by Ambainis in the lower bound proof of the usual permutation inversion problem~\cite{ambainis:adv}.
We need the following technical result, which was used~\cite{spalek:advEquivalent} to reduce the combinatorial formulation of the positive-weighted adversary like in~\cite{ambainis:adv} to the spectral formulation as in~\cite{barnum:advSpectral,hoyer:advNegative}.
We give a slightly modified version.

\begin{lem}
\label{lem:spalek}
Let $A$ be a matrix with entries $0,\pm1$.
Then,
\[
\|A\|\le \max_{i,j\colon A\elem[i,j]\ne 0} \sqrt{R_i C_j},
\]
where $R_i$ and $C_j$ is the number of non-zero elements in the $i$-th row and $j$-column, respectively.
\end{lem}

\pfstart
Taking the absolute value of each entry can only increase the norm, hence, we can assume the matrix $A$ only has 0,1 entries.
Then, this is a special case of Lemma 4.2 of~\cite{spalek:advEquivalent}.
\pfend

Assume $\k0\mapsto \ket|\tau_\pi>$ is a state-generating problem such that measuring $\tau_\pi$ gives $\pi^{-1}(1)$ with probability at least $2/3$.
We use this property to ensure that
\begin{equation}
\label{eqn:iptaupitausigma}
\mathrm{Re}\ip<\tau_\pi, \tau_\sigma> \le \frac{2\sqrt2}3
\qquad\text{for $\pi,\sigma\in\mathfrak S_n$ such that $\pi^{-1}(1)\ne\sigma^{-1}(1)$}.
\end{equation}
Define the corresponding output object, which is an $\mathfrak S_n\times \mathfrak S_n$-matrix $E$ with
\[
E\elem[\pi,\sigma] = 1-\ip<\tau_\pi, \tau_\sigma>.
\]

Let us define the adversary matrix $\Gamma$.
Denote by $\mathfrak C_n$ the subset of $\mathfrak S_n$ formed by permutations having a single cycle of length $n$.
We will only consider permutations in $\mathfrak C_n$.

\mycutecommand\lrs{\leftrightsquigarrow}
We say that $\pi, \sigma \in \mathfrak C_n$ are \emph{in relation}, denoted $\pi\lrs\sigma$, if $\pi$ and $\sigma$ have cyclic structures of the following form:
\begin{equation}
\label{eqn:relation}
\begin{aligned}
\pi\colon &1\mapsto \cdots\mapsto p_{k} \mapsto 
p_{k+1}\mapsto \cdots p_\ell \mapsto p_{\ell+1}\mapsto \cdots p_n \mapsto 1,\\
\sigma\colon &1\mapsto \cdots\mapsto p_{k} \mapsto 
p_{\ell+1}\mapsto \cdots p_n \mapsto p_{k+1}\mapsto \cdots p_\ell \mapsto 1.
\end{aligned}
\end{equation}
for some  $1\le k<\ell<n$.
In other words, the interval $p_{k+1}\mapsto \cdots\mapsto p_\ell$ is taken out and put at the end of the cycle.
Alternatively, one can say that the suffix $p_{k+1}\mapsto \cdots\mapsto p_n$ is cyclically shifted.
This is a symmetric relation, but neither reflexive, nor transitive.

As usual for the positive-weighted adversary, define an $\mathfrak C_n\times \mathfrak C_n$ matrix $\Gamma$ by $\Gamma\elem[\pi,\sigma] = 1_{\pi\lrs\sigma}$.

\begin{lem}
\label{lem:Gamma}
We have the following properties of the matrix $\Gamma$:
\begin{itemize}\itemsep0pt
\item $\Gamma$ is a Hermitian matrix;
\item $\Gamma\elem[\pi,\sigma]=0$ if $\pi^{-1}(1) = \sigma^{-1}(1)$;
\item $\lambda_{\max}(\Gamma) = \Omega(n^2)$ with the principal eigenvector given by the all-1 vector;
\item $\lambda_{\max}(-\Gamma) \le n-2$.
\end{itemize}
\end{lem}

\pfstart
The first two properties follow from the definition of the relation:
If $\pi\lrs\sigma$, then $\sigma\lrs\pi$.
Also, in this case, $\pi^{-1}(1) \ne \sigma^{-1}(1)$.
The third property follows from the fact that each row has exactly $(n-1)(n-2)/2$ ones.

Now, let us prove the fourth property.
It is equivalent to $(n-2)I + \Gamma\succeq 0$.
Let us prove the latter.
Fix $1\le k\le n-2$.
Say that $\pi \sim_k \sigma$ if $\pi = \sigma$ or $\pi$ and $\sigma$ are in relation like in~\rf{eqn:relation} with this fixed value of $k$.
Note that $\sim_k$ is an equivalence relation.
Define the matrix $\Gamma_k$ by $\Gamma_k\elem[\pi,\sigma] = 1_{\pi\sim_k \sigma}$.
It is a block-diagonal matrix with all-1 blocks on the diagonal.
Hence, $\Gamma_k\succeq 0$, which gives
\[
\sum_{k=1}^{n-2} \Gamma_k  = (n-2)I + \Gamma \succeq 0.\qedhere
\]
\pfend

Let $u$ be the normalised all-1 vector.  Then,
\begin{equation}
\label{eqn:permutationLower}
\lambda_{\max}(\Gamma\circ E)
\ge u^*(\Gamma\circ E) u
\ge \sC[1-\frac{2\sqrt2}3] u^*\Gamma u = \Omega(n^2),
\end{equation}
where we used the second point of \rf{lem:Gamma} and~\rf{eqn:iptaupitausigma} on the second step, and the third point of \rf{lem:Gamma} on the third.

It remains to estimate $\Gamma\circ \Delta$, where
\[
\Delta_{\pi,\sigma} = I - O_\pi^*O_\sigma = I-O_{\pi^{-1}\sigma}.
\]
By the definition of $\Gamma$, we can restrict our attention to the pairs $\pi,\sigma$, which are in relation~\rf{eqn:relation}.
In notation of~\rf{eqn:relation}, we have that $\pi^{-1}\sigma$ is a single cycle of length 3
\[
\pi^{-1}\sigma\colon p_k\mapsto p_\ell\mapsto p_n \mapsto p_k
\]
and identity elsewhere.
Hence,
\begin{equation}
\label{eqn:Deltapisigma}
\Delta_{\pi,\sigma} = 
\begin{pmatrix}
1 & 0 & -1\\
-1& 1 & 0\\
0 & -1& 1\\
\end{pmatrix}
\end{equation}
where the rows and columns are labelled by $p_k,p_\ell,p_n$ in this order,
and the matrix has zeroes everywhere else.

The matrix $\Gamma\circ\Delta$ is labelled by the elements $(\pi,i)\in\mathfrak C_n\times[n]$.
The block~\rf{eqn:Deltapisigma} when embedded in the latter has 
\[
\text{rows $(\pi, p_k), (\pi, p_\ell) , (\pi, p_n)$ \qquad and \qquad columns $(\sigma, p_k), (\sigma, p_\ell), (\sigma, p_n)$}.
\]

We would like to apply \rf{lem:spalek} to $\Gamma\circ\Delta$.  For instance, we see that there are at most $n$ choices of $\rho\in\mathfrak C_n$ such that $\Delta_{\pi,\rho}$ has non-zero elements in row $(\pi, p_k)$, since $p_k$ has to be one of the two elements used to define the relation $\pi\lrs \rho$ for this to happen.
Similarly, in notation of \rf{lem:spalek}, we get the following estimates:
\begin{subequations}
\label{eqn:RC}
\begin{equation}
R_{\pi, p_k}, R_{\pi, p_\ell}, C_{\sigma, p_k}, C_{\sigma, p_n} \le 2n.
\end{equation}
For one row and one column we get a worse estimate, where we count the total number of $\rho$ in relation with $\pi$ (or $\sigma$, respectively):
\begin{equation}
R_{\pi, p_n}, C_{\sigma, p_\ell} \le 2n^2.
\end{equation}
\end{subequations}
Therefore, we should treat the element on the intersection of the latter row and the latter column separately.
Rewrite~\rf{eqn:Deltapisigma}:
\[
\Delta_{\pi,\sigma} = 
\begin{pmatrix}
1 & 0 & -1\\
-1& 1 & 0\\
0 & -1& 1\\
\end{pmatrix}
=
\begin{pmatrix}
1 & 0 & -1\\
-1& 1 & 0\\
0 & 0 & 1\\
\end{pmatrix}
+
\begin{pmatrix}
0 & 0 & 0\\
0 & 0 & 0\\
0 & -1& 0\\
\end{pmatrix}
,
\]
Let us denote the first and the second matrices in the last sum by $\Delta_{\pi,\sigma}'$ and $\Delta_{\pi,\sigma}''$, respectively, and the corresponding families by $\Delta'$ and $\Delta''$.

\begin{clm}
\label{clm:Delta'}
$\|\Gamma\circ\Delta'\| = O(n^{3/2})$.
\end{clm}

\pfstart
This follows from \rf{lem:spalek} using the estimates in~\rf{eqn:RC}.
\pfend

\begin{clm}
\label{clm:Delta''}
$\lambda_{\max}(\Gamma\circ\Delta'') = O(n)$.
\end{clm}

\pfstart
The matrix $\Gamma\circ\Delta''$ is just the matrix $-\Gamma$ where the row and the column label $\pi$ becomes $\sA[\pi, \pi^{-1}(1)]$ and the matrix is extended by zeroes elsewhere.
Hence, the claim follows from the fourth point of \rf{lem:Gamma}.
\pfend

Combining Claims~\ref{clm:Delta'} and~\ref{clm:Delta''}, we get that
\[
\lambda_{\max} (\Gamma\circ\Delta) = O(n^{3/2}).
\]
Together with~\rf{eqn:permutationLower}, this gives the required lower bound.

\section{Discussion and Future Work}
\label{sec:discussion}
In this paper, we defined a natural notion of Las Vegas complexity, and demonstrated its versatility for various composition results.
We proved that a Las Vegas algorithm can be turned into an approximate Monte Carlo algorithm with a slight increase in complexity.
We have shown that Las Vegas complexity is equal to the adversary bound for exact state and subspace conversion.

The latter is exciting as the same object is shown to have two different facets.
For some problems, intuition gathered from quantum algorithms might be helpful in coming up with good Las Vegas algorithms.
For other problems, it might be easier to forget about limitations of quantum algorithms and work directly with optimisation problems.
Our algorithm of~\rf{sec:mainupper} can be seen as a way to guess arbitrarily large states to process by the input oracle.
It is interesting to understand consequences of such a subroutine.

Due to its exactness, Las Vegas complexity results in ``cleaner'' algorithms without necessity to worry about error reduction.
Similar results have been already obtained for function evaluation using compositional properties of the adversary bound.
But having ``clean'' subroutines for various state-generating and state-converting problems might be helpful as well, especially, given that they not always have built-in tools for error reduction.

This paper should be seen as a prequel to Ref.~\cite{belovs:variations} since it gives a more general and simple exposition of the first half of Ref.~\cite{belovs:variations}, but mostly ignores the second half, which deals with applications to function and relation evaluation.
Complete reconciliation of the results from~\cite{belovs:variations} with the current paper is left as important future work.
Let us just mention two results that can be easily obtained in this way.
\itemstart
\item Purifiers of~\cite{belovs:variations} imply that, assuming bidirectional access, a Monte Carlo algorithm for approximate and non-coherent function evaluation can be turned into an exact coherent Las Vegas algorithm for the same function with constant increase in complexity.
As mentioned in the introduction, this is in contrast to randomised Las Vegas complexity.
\item The bound~\rf{eqn:BooleanTwo} is equal to Las Vegas complexity of bidirectional function evaluation also for non-Boolean functions.
However, we only get this result up to a constant factor.
Understanding the exact relation between the two is still an open problem.
\itemend

We list just a few other open problems.
What is Las Vegas complexity of various important subroutines, for instance, amplitude amplification?
Can the adversary bound for contraction oracles be applied for some problems like faulty oracles?
Is there an nice formulation of the adversary bound for (approximate and non-coherent) function evaluation with unidirectional input oracles?
In particular, what is the true complexity of the unitary permutation inversion problem?
The purifiers mentioned above seem to crucially depend on bidirectionality.

Finally, an interesting research direction is to obtain analogues of some of the results in this paper for time complexity.

\subsection*{Acknowledgements}
We thank anonymous reviewers for their comments on an earlier draft of the paper.

A.B. is supported by the ERDF project number 1.1.1.5/18/A/020 ``Quantum algorithms: from complexity theory to experiment''.

\bibliographystyle{habbrvM}
{
\small
\bibliography{belov}
}

\appendix

\section{Duality}
\label{app:duality}
We use semi-definite duality.
The dual is constructed by explicitly writing down the Lagrangian and transforming it.  
Thus, \emph{weak duality} (the maximisation problem bounds the minimisation problem from below) is apparent.  
To prove \emph{strong duality} (their optimal values are equal), we rely on Slater's condition.  The latter says that strong duality holds if one of the optimisation problems is convex and \emph{strictly feasible}, i.e. there exists a feasible solution making all the inequalities in the problem strict.

It turns out that the calculations are concise using multidimensional tensors with contractions given by the inner product formula between Hermitian matrices: $\ip<A,B> = \tr A^*B$.
However, given that matrices are tensors themselves, this notation might be confusing, so we opted to use the following one, that we find more intuitive.

We assume the matrices are square and are labelled by elements of direct products of some sets.  If $A$ is a matrix labelled by $X\times Y$, and $B$ is a matrix labelled by $X\times Z$, then $A\circ B$ is a matrix labelled by $X\times Y\times Z$ given by
\[
A\circ B\elem[(x,y,z),(x',y',z')] = A\elem[(x,y),(x',y')]\; B\elem[(x,z),(x',z')].
\]
This includes the usual Hadamard product (when $|Y|=|Z|=1$), the tensor product (when $|X|=1$) and the version of the Hadamard product used in~\rf{eqn:onegammadualCondition} (when $|Y|=1$).

\newcommand{\suml}{\mathop{\textstyle\sum}\nolimits}

For the matrix $A$ as above, let $\suml_Y A$ be the $X\times X$ matrix given by
\[
\sA[\suml_Y A]\elem[x,x'] = \sum\nolimits_{y,y'\in Y} A\elem[(x,y), (x',y')].
\]
$\suml$ without the subindex stands for the total sum of all entries.
In particular, we have $\ip<A,A'> = \suml(\overline A\circ A')$ and the partial trace is $\tr_Y(A) = \suml_Y(A\circ I_Y)$, where $\overline A$ is complex conjugate and $I_Y$ is the $Y\times Y$ identity matrix.

\pfstart[Proof of \rf{thm:onegammadual}]
We have three sets of labels: $D$, and the bases of $\cM$ and $\cW$, for which we use letters $M$ and $W$.
By~\rf{eqn:onegammaCondition}:
\[
e_{xy} 
= \tr \skA[v_x^*(\Delta_{xy}\otimes I_W) v_y]
= \tr \skA[v_yv_x^*(\Delta_{xy}\otimes I_W)]
= \tr \skA[(v_xv_y^*)^*(\Delta_{xy}\otimes I_W)]
= \suml\sA[\overline{v_xv_y^*} \circ \Delta_{xy}\circ I_W].
\]
Let us merge all these conditions into one.
Let $E$ be the $D\times D$ matrix given by $(e_{xy})$, and
$\Delta$ be the $(D\times M)\times (D\times M)$ matrix with the blocks $\Delta_{x,y}$.
Both these matrices are Hermitian.
Let also $v$ be the vector in $\bC^{D\times M\times W}$ obtained by joining all $v_x$.
Then, all the constraints in~\rf{eqn:onegammaCondition} can be concisely written as
\[
E = \suml_{M,W} (\overline{vv^*} \circ \Delta \circ I_W)
= \suml_M\sA[\suml_W (\overline{vv^*}\circ I_W)\circ \Delta ] = \suml_M\sA[X\circ \Delta],
\]
where $X$ is a positive semi-definite $(D\times M)\times (D\times M)$-matrix given by $X = \tr_W(\overline{vv^*})$.
Conversely, any positive semi-definite matrix can be written in this way for a large enough $W$.
Also, the matrix $\sum_M(X\circ I_{D,M})$ is the diagonal matrix with $\|v_x\|^2$ on the diagonal.
Therefore, we get the following equivalent formulation of the optimisation problem~\rf{eqn:onegamma}:
\begin{subequations}
\label{eqn:alternativeonegamma}
\begin{alignat}{2}
&\mbox{\rm minimise} &\quad& t \\
& \mbox{\rm subject to}&&  
tI_D\succeq \suml_M(X\circ I_{D,M})  \\
&&&E = \suml_M(X\circ \Delta)\\
&&&X\succeq 0,\quad t\in \bR.
\end{alignat}
\end{subequations}

We introduce two Lagrangian multipliers $Y\succeq 0$ and $\Lambda$ which are $D\times D$ Hermitian matrices, resulting in the following Lagrangian:
\[
t + \suml_D\skB[ Y\circ {\sA[\suml_M (X\circ I_{D,M}) - t I_D]}] + \suml_D\skB[ \Lambda\circ {\sA[E - \suml_M(X\circ \Delta)]}]
\]
After rearrangement:
\[
\suml_D (\Lambda\circ E) + t\skA[1-\tr Y] + \suml_{D,M} \skA[X\circ(Y\circ I_{D,M} - \Lambda\circ\Delta)]
\]
This gives the following dual:
\begin{subequations}
\label{eqn:dual}
\begin{alignat}{2}
&\mbox{\rm maximise} &\quad& \suml_D(\Lambda\circ E) \label{eqn:alternativeonegammaObjective}\\
& \mbox{\rm subject to}&&  
\tr Y = 1 \label{eqn:trY}\\
&&& \Lambda\circ\Delta \preceq  Y\circ I_{D,M} \label{eqn:LambdacircDelta}\\
&&&Y\succeq 0,\quad\text{$\Lambda$ Hermitian}.
\end{alignat}
\end{subequations}
Note that this optimisation problem is strictly feasible as it suffices to take $\Lambda=0$ and $Y$ a multiple of the identity matrix satisfying $\tr Y=1$.
Therefore, by Slater's condition, the optimal values of~\rf{eqn:alternativeonegamma} and~\rf{eqn:dual} are equal.

By studying~\rf{eqn:LambdacircDelta}, we see that we can assume that $Y$ is rank-1 (by extending the diagonal matrix $Y\circ I_{D,M}$), and we can write $\Lambda$ as $\Gamma\circ Y$ for some Hermitian $D\times D$-matrix $\Gamma$.
Then~\rf{eqn:LambdacircDelta} becomes
\begin{equation}
\label{eqn:YcircGammacircDelta}
Y\circ\Gamma\circ\Delta \preceq  Y\circ I_{D,M},
\end{equation}
and the objective~\rf{eqn:alternativeonegammaObjective} becomes
\begin{equation}
\label{eqn:alternativeonegammaObjective2}
\suml_D(Y\circ\Gamma\circ E).
\end{equation}
This is clearly continuous in $Y$ for fixed $\Gamma$ and $E$, thus, we can additionally assume that $Y$ has non-zero diagonal.
Then, the Hadamard inverse of $Y$ is defined and positive semi-definite, hence, Eq.~\rf{eqn:YcircGammacircDelta} is equivalent to
\begin{equation}
\label{eqn:GammacircDelta}
\Gamma\circ\Delta \preceq I_{D,M}.
\end{equation}
Altogether, the objective~\rf{eqn:alternativeonegammaObjective2} with conditions~\rf{eqn:GammacircDelta}, \rf{eqn:trY}, and $Y$ is positive semi-definite rank-1 gives us the dual
\begin{alignat*}{2}
&\mbox{\rm maximise} &\quad& \lambda_{\max} (\Gamma\circ E) \\
& \mbox{\rm subject to}&&  \lambda_{\max}(\Gamma\circ \Delta)\le 1 
\end{alignat*}
as required.
\pfend

\begin{proof}[Proof of \rf{clm:feasibleClosed}]
The set of feasible solutions of the optimisation problems~\rf{eqn:onegammaMultidimensional} and~\rf{eqn:onegammaMultiobjective} is the same so we can use the same characterisation~\rf{eqn:alternativeonegamma} as in the proof of \rf{thm:onegammadual}.

Let $W$ denote the feasible objective space of the optimisation problem, 
and $B_R$ denote the set of vectors in $\bR^D\otimes \bR^s$ with the sum of entries bounded by $R$.
The objective profile $w = (\Dnorm|v_x|)_{x\in D}$ can be obtained by summing the diagonal entries of the corresponding matrix $X$.
Thus, $W\cap B_R$ is the image under a continuous map of the set of feasible solutions $X$ to~\rf{eqn:alternativeonegamma} with $\tr X\le R$.
The latter set is easily seen to be compact, hence, $W\cap B_R$ is closed.
As $R$ is arbitrary, $W$ is closed as well.
\end{proof}

\end{document}